\documentclass[12pt]{article}

\usepackage{amsfonts,amscd,amssymb}
\usepackage{amsthm,amsmath} 
\usepackage{bbm} %bb font numbers
\usepackage{times,graphicx}
\usepackage{mathrsfs}
\usepackage{enumerate}
\usepackage{subfigure}
\usepackage{appendix}
\usepackage{pdfpages}
\usepackage{color}
\usepackage[margin=1in]{geometry}

\usepackage{tikz}
\usetikzlibrary{arrows,decorations.pathreplacing}

\definecolor{Red}{rgb}{1,0,0}
\definecolor{Blue}{rgb}{0,0,1}
\definecolor{Olive}{rgb}{0.41,0.55,0.13}
\definecolor{Green}{rgb}{0,1,0}
\definecolor{MGreen}{rgb}{0,0.8,0}
\definecolor{DGreen}{rgb}{0,0.55,0}
\definecolor{Yellow}{rgb}{1,1,0}
\definecolor{Cyan}{rgb}{0,1,1}
\definecolor{Magenta}{rgb}{1,0,1}
\definecolor{Orange}{rgb}{1,.5,0}
\definecolor{Violet}{rgb}{.5,0,.5}
\definecolor{Purple}{rgb}{.75,0,.25}
\definecolor{Brown}{rgb}{.75,.5,.25}
\definecolor{Grey}{rgb}{.5,.5,.5}
\definecolor{Black}{rgb}{0,0,0}

\RequirePackage[colorlinks,citecolor=blue,urlcolor=blue]{hyperref}
\newtheorem{thm}{Theorem}
\newtheorem{lemma}{Lemma}
\newtheorem{prop}{Proposition}
\newtheorem{cor}{Corollary}

\newtheorem{example}{Example}

\newtheorem{definition}{Definition}
\newtheorem{assumption}{Assumption}

\theoremstyle{empty}

\newcommand{\one}{\mathbf{1}}

\newcommand{\Ef}{{\mathrm{I}\!\mathrm{E}}}
\newcommand{\E}[1]{\Ef\!\left[#1\right]}
\newcommand{\Pf}{{\mathrm{I}\!\mathrm{P}}}
\renewcommand{\P}[1]{\Pf\!\left[#1\right]}

\newcommand{\tree}{\mathbb{T}}
\newcommand{\leaves}{\mathscr{L}}
\newcommand{\internal}{\mathscr{I}}
\newcommand{\var}{\mathrm{Var}}
\newcommand{\evar}{\widehat{\mathrm{Var}}}
\newcommand{\cov}{\mathrm{Cov}}
\newcommand{\mle}{\hat{\mu}_{\mathrm{ML}}}
\newcommand{\bftheta}{{\boldsymbol\theta}}
\newcommand{\bfeta}{{\boldsymbol\eta}}
\newcommand{\gr}{\Lambda^\mathrm{g}}
\newcommand{\unigr}{\Lambda^\mathrm{ug}}
\newcommand{\ugr}{\overline{\Lambda}^\mathrm{g}}
\newcommand{\lgr}{\underline{\Lambda}^\mathrm{g}}
\newcommand{\br}{\Lambda^\mathrm{b}}
\newcommand{\treeseq}{\mathcal{T}}
\newcommand{\ind}{\mathbbm{1}}
\newcommand{\ical}{\mathcal{I}}
\newcommand{\jcal}{\mathcal{J}}
\newcommand{\neff}{n^{\mathrm{eff}}}
\newcommand{\contrast}{\mathcal{C}}
\newcommand{\contrasts}{\mathscr{C}}

\title{Phase transition on the convergence rate of 
parameter estimation under an Ornstein-Uhlenbeck diffusion on a tree}
\date{}
\author{
C\'ecile An\'e\footnote{Departments of Statistics and of Botany, 
University of Wisconsin-Madison.
Work supported by NSF grants DMS-1106483.}
\and
Lam Si Tung Ho\footnote{Departments of Statistics, University of Wisconsin-Madison.}
\and
Sebastien Roch\footnote{Departments of Mathematics and Statistics (by courtesy), 
University of Wisconsin-Madison.
Work supported by NSF grants DMS-1007144 and DMS-1149312 (CAREER), and an Alfred P. Sloan Research Fellowship.}
}

\begin{document}
\maketitle

\begin{abstract}
Diffusion processes on trees are commonly used in evolutionary biology to model the joint distribution of continuous traits, such as body mass, across species.
Estimating the parameters of such processes 
from tip values
presents challenges because of the intrinsic correlation
between the observations produced by the shared evolutionary 
history, thus violating the standard independence assumption 
of large-sample theory.
For instance Ho and An\'e \cite{HoAne13} recently proved 
that the mean (also known in this context as selection optimum) 
of an Ornstein-Uhlenbeck process on a tree 
cannot be estimated consistently from
an increasing number of tip observations 
if the tree height is bounded. 
Here, using a fruitful connection to the so-called
reconstruction problem in probability theory,
we study the convergence rate of parameter estimation
in the unbounded height case. 
For the mean of the process, 
we provide a necessary and 
sufficient condition for the consistency of the maximum 
likelihood estimator (MLE) 
and establish a phase transition on its convergence rate
in terms of the growth of the tree. In particular we show that
a loss of $\sqrt{n}$-consistency
(i.e., the variance of the MLE becomes $\Omega(n^{-1})$,
where $n$ is the number of tips) 
occurs when the tree growth is larger than a threshold 
related to the phase transition of the reconstruction problem.
For the covariance parameters, we give a novel, efficient 
estimation method which achieves 
$\sqrt{n}$-consistency under natural assumptions on the tree.
Our theoretical results provide practical suggestions for the 
design of comparative data collection.

\paragraph{Keywords}
Ornstein-Uhlenbeck, phase transition, evolution, phylogenetic, consistency, maximum likelihood estimator.
\end{abstract}

%\section{Introduction}
%\label{intro}
%Your text comes here. Separate text sections with
%\section{Section title}
%\label{sec:1}
%Text with citations \cite{RefB} and \cite{RefJ}.
%\subsection{Subsection title}
%\label{sec:2}
%as required. Don't forget to give each section
%and subsection a unique label (see Sect.~\ref{sec:1}).
%\paragraph{Paragraph headings} Use paragraph headings as needed.
%\begin{equation}
%\end{equation}

%\begin{figure}
%  \includegraphics{example.eps}
%\end{figure}
%\begin{figure*}
%  \includegraphics[width=0.75\textwidth]{example.eps}
%\end{figure*}
%\begin{table}
%\begin{tabular}{lll}
%\hline\noalign{\smallskip}
%\end{tabular}
%\end{table}

\section{Introduction}

Analysis of data collected from multiple species presents
challenges because of the intrinsic correlation produced by the shared
evolutionary history.
This dependency structure can be modeled by assuming that the
traits of interest evolved along a phylogeny according to a stochastic process.
Two commonly used processes for continuous traits,
such as body mass, are Brownian motion (BM) and 
the Ornstein-Uhlenbeck (OU) process.
BM is used to model neutral evolution, with no favored direction
(see e.g. \cite{felsenstein1985phylogenies}).
On the other hand, the OU process can account for natural selection using two 
extra parameters: a ``selection optimum'' $\mu$ towards which the process is attracted
and a ``selection strength'' $\alpha$ \cite{hansen1997stabilizing}. 
The OU process has a stationary distribution, which is Gaussian with mean 
$\mu$ and variance $\gamma = {\sigma^2/2 \alpha}$. 
The presence of natural selection can be detected by testing 
whether $\alpha>0$ (e.g. \cite{harmon-etal10}).
Changes in $\mu$ across different groups of organisms are
used to correlate changes in selection regime with
changes in behavior or environmental conditions
(see e.g. \cite{butler2004phylogenetic,bartoszek2012multivariate}). 
For instance, the optimal body size $\mu$ might be different 
for terrestrial animals than for birds and bats.
In practice, $\mu$, $\alpha$ and the infinitesimal variance $\sigma^2$ 
(or stationary variance $\gamma$) are estimated from data on extant species. In other
words, only data at the tips of the tree are available. The process at internal nodes
and edges is unobserved. Also, the tree is reconstructed independently from 
external and abundant data, typically from DNA sequences. 
In practice there can be some uncertainty about a few nodes in the tree, 
but we assume here that the tree is known without error. 
 
The OU process on a tree has been used extensively in practice
(see e.g. \cite{butler2004phylogenetic,cooperPurvis10,brawand_etal11,rohlfs2014modeling}), 
but very few authors have studied 
convergence rates of available estimators.
Recently Ho and An\'e \cite{HoAne13} 
showed that if the tree height is bounded as the sample size goes to infinity, 
no estimator for $\mu$ can ever be consistent. This is because $\mu$ is not
``microergodic": the distribution $P_\mu$ of the whole observable process
$(Y_i)_{i\geq 1}$ at the tips of the tree is such that $P_{\mu_1}$ and $P_{\mu_2}$ 
are not orthogonal for any values $\mu_1\neq\mu_2$, if the tree height is bounded.
This boundedness assumption does not hold for 
common models of evolutionary trees however, 
such as the pure-birth (Yule) process \cite{yule1925mathematical}. 
We consider here the case of an unbounded tree height.
We study the consistency and convergence rates of several 
estimators, including some novel estimators, using
tools from the literature on the reconstruction problem
in probability theory. In particular we relate the convergence
rates of these estimators to the growth rate of the phylogeny.
This connection is natural given that the growth rate 
(and the related branching number) is known to play an 
important role in the analysis of a variety of stochastic
processes including random walks, 
percolation and ancestral state reconstruction on trees~\cite{peres1999climb}. 
In particular we leverage a useful
characterization of the variance of linear estimators
in terms of electrical networks.
 
\paragraph{Main results} 
We present the asymptotic properties of two common estimators for $\mu$: 
the sample mean and the maximum likelihood estimator (MLE). Conditional on the tree, 
the MLE $\mle$ is known to be the best linear unbiased 
estimator for $\mu$ assuming that $\alpha$ is known. 
(The assumption of known $\alpha$ is proved not to be restrictive for our convergence rate results if $\alpha$ can be well estimated.)
In fact, we give an example when $\mle$ performs significantly better than the sample mean, which is not consistent in that
particular case.
In one of our main results, we identify a necessary 
and sufficient condition for the consistency of $\mle$. 
We also derive a phase transition on its convergence rate, 
which drops from $\sqrt{n}$-consistency (i.e. the variance is
$O(n^{-1})$)
to a lower rate, $n$ being the number of samples (i.e. tip observations). 
This phase transition depends on the growth rate of the tree.
Tree growth measures the rate at which new leaves arise 
as the tree height increases
(see Section \ref{sec:def} for a formal definition).
Roughly, 
when the growth rate is below $2\alpha$, 
we show that $\sqrt{n}$-consistency holds.
This is intuitive as a lower growth rate means
lower correlations between the leaf states.
On the other hand, when the growth rate is above $2\alpha$
implying a sample size $n \gg e^{2\alpha T}$,
i.e. when the tree is sufficiently ``bushy,''
then the ``effective sample size" is reduced to $\neff = e^{2\alpha T}$ and 
the $\sqrt n$-consistency of $\mle$ is lost.
We also provide novel, efficient estimators for the other two parameters, 
$\alpha$ and $\gamma$, which achieve $\sqrt{n}$-consistency and 
do not require the knowledge of $\mu$. Interestingly, the 
$\sqrt{n}$-consistency in this case is not affected 
by growth rate, unlike the case of the MLE for $\mu$.
Our results lead to a practical method to assess whether additional species are informative or not, thus helping researchers to avoid wasting money and effort. 
Section~\ref{sec:app} presents simulations to illustrate these suggestions.
Our main results are stated formally and further discussed in Section \ref{sec:def},
after necessary definitions. Their proofs are found in Section \ref{sec:proofs}.

\paragraph{Related work}
Bartoszek and Sagitov \cite{bartoszek2012phylogenetic} obtained a corresponding phase transition for 
the convergence rate of the sample mean to estimate $\mu$, assuming a Yule process for the tree.
%with unit speciation rate. They proved that the sample mean is $\sqrt{n}$-consistent when
% $\alpha > 0.5$ and is not $\sqrt{n}$-consistent  when $\alpha \leq 0.5$. 
Phase transitions for the convergence rate of some U-statistics have also been obtained for the 
OU model when the tree follows a supercritical branching process \cite{adamczak2011clt,adamczak2011u}. 
A main difference between these studies and our work is that 
we assume that the tree is known.
Even though tree-free estimators are the only practical options when the tree is unknown,
this situation is now becoming rare due to the ever-growing availability of sequence data for building trees. For instance 
Crawford and Suchard \cite{crawford2013diversity} acknowledge that ``as evolutionary biologists 
further refine our knowledge of the tree of life, the number of clades whose phylogeny is truly 
unknown may diminish, along with interest in tree-free estimation methods.'' 

As we mentioned, related phase transitions have been
obtained for other processes on trees.
For instance, the growth rate of the tree determines whether the state 
at the root can be reconstructed better than random for a binary symmetric channel on a 
binary tree (see e.g. \cite{EvKePeSc:00} and references therein).
In a recent result, Mossel and Steel~\cite{mossel2014majority}
established a transition for ancestral state reconstruction
by majority rule 
for the binary symmetric model on a Yule tree 
at the same critical point as above. 
Note that majority rule is a tree-free estimator
like the sample mean in \cite{bartoszek2012phylogenetic}, 
but adapted to discrete traits.
In the context of the OU model,
Mossel et al. \cite{mossel2013robust} obtained a phase 
transition for estimating the ancestral state at the root,
with the same critical growth rate we derive in our results.

\section{Definitions and statements of results}
\label{sec:def}

In this section, we state formally and further explain
our main results. First, we define our model and
describe the setting in which our results are proved.

\subsection{Model}

Our main model is a stochastic process on a species tree
$\tree$. 
Let $\tree = (\mathscr{E},\mathscr{V})$ be a finite
tree with leaf set $\leaves = \{1,\ldots,n\}$ and root $\rho$.
The leaves typically correspond to extant species.
We think of the edges of $\tree$
as being oriented away from the root.
To each edge (or branch) $b \in \mathscr{E}$ of the tree is associated
a positive length $|b| > 0$ corresponding to the time
elapsed between the endpoints of $b$. 
For any two vertices $u,v \in \mathscr{V}$, we denote by 
$d_{uv}$ the distance between $u$ and $v$ in
$\tree$, that is, the sum of the branch lengths
on the unique path between $u$ and $v$.
We assume that the species tree is {\em ultrametric}, that is, 
that the distance from the root to every leaf is the same. 
It implies that, for any
two tips $i,j \in \leaves$, $d_{ij}$ is twice the time to the most
recent common ancestor of $i$ and $j$ from the leaves.
We let
$T$ be the height of $\tree$, that is, the distance
between the root and any leaf, and we define
$t_{ij} = T - \frac{d_{ij}}{2}$.
{\em Throughout we assume
that the species tree is known.}

We consider an Ornstein-Uhlenbeck (OU) process
on $\tree$. That is, on each branch of $\tree$, we 
have a diffusion
$$
dY_t = -\alpha(Y_t - \mu)dt + \sigma dB_t,
$$
where $B_t$ is a standard Brownian motion (BM).
In the literature on continuous traits,
$Y_t$ is known as the response variable,
$\mu$ is the selection optimum,
$\alpha > 0$ is the selection strength,
$\sigma > 0$ is the scale parameter of the Brownian motion.
We assume that the root value follows
the stationary Gaussian distribution ${\cal N}(\mu,\gamma)$,
where $\gamma = \frac{\sigma^2}{2 \alpha}$.
At each branching point, we run the process
independently on each descendant edge
starting from the value at the
branching. 
Equivalently, the column vector of observations  $\mathbf{Y} = (Y_\ell)_{\ell \in \leaves}$ at the tips
of the tree are Gaussian with mean $\mu$ and variance matrix 
$\mathbf{\Sigma} = \gamma \mathbf{V_\tree}$ where
\begin{equation*}
%\label{eq:randomrootV}
(V_\tree)_{ij}=e^{-\alpha d_{ij}}.
\end{equation*} 
{\em We assume throughout that $\alpha$,
$\mu$ and $\sigma$ are the same on every branch
of $\tree$. We will specify below whether these
parameters are known, depending on the context.}

\paragraph{Parameter estimators}
Our interest lies in estimating the parameters
of the model, given $\tree$, from a sample
of $\mathbf{Y}$. In addition to proposing new
estimators for $\alpha$ and $\sigma$, we 
study common estimators of $\mu$. 
In particular we consider the empirical average at the tips
$
\overline{Y} = \one'\mathbf{Y} /n,
$ % \frac{1}{n} \sum_{\ell \in \leaves} Y_\ell.
where $\one$ denotes the all-ones vector % $\zero$ for all-zeros, with the size dictated by the context
and $\mathbf{v}'$ denotes the transposes of a vector or matrix $\mathbf{v}$.
Also, the MLE of $\mu$ {\em given the tree and $\alpha$}
is
\[
\mle = (\one' \mathbf{V}_\tree^{-1} \one)^{-1}  \one' \mathbf{V}_\tree^{-1} \mathbf{Y},
\]
which is the well-known generalized least squares 
estimator for the linear regression problem
$\mathbf{Y} =  \mu \one + \boldsymbol{\varepsilon},$
where $\boldsymbol{\varepsilon}$ is multivariate normal with 
covariance matrix $\mathbf{\Sigma}$ (see e.g.~\cite{Anderson:1984}).
Note that the mean squared error is given by
\begin{equation}
\var_\tree[\mle] 
= (\one' \mathbf{V}_\tree^{-1} \one)^{-2} \one' \mathbf{V}_\tree^{-1}\mathbf{\Sigma} (\mathbf{V}_\tree^{-1})' \one
= \gamma (\one' \mathbf{V}_\tree^{-1} \one)^{-1}.
\label{eq:varmuhat}
\end{equation}
We drop the $\tree$ in $\var_\tree$ when the tree
is clear from the context.

The estimators $\overline{Y}$ and $\mle$ are both linear
estimators. 
% More generally, for any weight vector $\bftheta = (\theta_\ell)_{\ell\in\leaves}$
% with $\bftheta' \one  = 1$, the following
% $$Y_\bftheta = \sum_{\ell \in \leaves} \theta_\ell Y_\ell,$$
% is an unbiased estimator of $\mu$. 
It is
useful to think of the MLE in this context 
as an unbiased linear estimator
minimizing the mean squared error (that is, a best linear unbiased estimator), which follows
from the Gauss-Markov Theorem \cite{shao2003mathstat}. 
% Indeed, for any $\bftheta = (\theta_\ell)_{\ell\in\leaves}$
% with $\bftheta' \one  = 1$,
% $$\var[Y_\bftheta]= \bftheta'\mathbf{\Sigma}\bftheta,$$
% which is minimized if $\mathbf{\Sigma}\bftheta = \lambda \one$
% for some $\lambda$, which leads to the optimal choice
% $$\bftheta^\star = (\one'\mathbf{\Sigma}^{-1}\one)^{-1} \mathbf{\Sigma}^{-1} \one$$
% and $Y_{\bftheta^\star}=\mle$.
% \begin{lemma}[Variational formulation]
% \label{lemma:variational-formulation}
% The MLE of $\mu$ given $\alpha$ and $\tree$
% minimizes the mean squared error
% among all linear unbiased estimators.
% \end{lemma}

\subsection{Asymptotic setting}

Our results are asymptotic. 
Specifically, we consider 
sequences of trees $\treeseq = (\tree_k)_{k \geq 1}$
with {\em fixed} parameters $\alpha, \mu, \sigma$. 
For $k \geq 1$, let $n_k$ be the number of 
leaves in $\tree_k$ and $T_k$ be the height
of $\tree_k$. As before, we denote the leaf set
of $\tree_k$ as $\leaves_k = [n_k]$. 
\begin{assumption}[Unboundedness]
\label{assumption:unbounded}
Throughout we assume
that $n_k \leq n_{k+1}$, $T_k \leq T_{k+1}$, 
and that $n_k \to +\infty$ and $T_k \to +\infty$
as $k \to +\infty$.
\end{assumption}

\noindent
For such a sequence of trees and a corresponding sequence
of estimators, say $X_k$, we recall various desirable asymptotic 
properties of $X_k$.
% \begin{definition}[Notions of convergence]
% We say that $(X_k)_{k\geq 1}$ {\em converges in probability} to $X$ if
% \[\forall \epsilon > 0,~ \lim_{k \to \infty} \P{|X_k - X| \geq \epsilon} = 0.\] 
% We denote this as $|X_k-X| = o_p(1)$.
% Moreover, we write $|X_k - X| = o_p(a_k)$ if $a_k^{-1}|X_k - X| = o_p(1)$.
% We say that $(X_k)_{k}$ is {\em bounded in probability} if
% for any $\delta > 0$, there exists $M_\delta > 0$ such that 
% \[\forall \delta > 0,~ \sup_{k} \P{|X_k| \geq M_\delta} <\delta.\]
% We denote this as $|X_k| = O_p(1)$. 
% Moreover, we write $|X_k| = O_p(a_k)$ if $a_k^{-1}|X_k| = O_p(1)$.
% \end{definition}
\begin{definition}[Consistency]
Let $(X_k)_k$ be a sequence of estimators for a parameter $x$. 
We say that $(X_k)_k$ is {\em consistent for $x$} if 
$X_k$ converges in probability to $x$, denoted as
$|X_k - x| = o_p(1)$. For $\beta > 0$,
we say that $(X_k)_k$ is {\em $(n_k^{\beta})$-consistent for $x$} if 
$(n_k^{\beta}(X_k-x))_k$ is bounded in probability, which we denote as
$|X_k - x| = O_p(n_k^{-\beta})$.  
\end{definition}
%\begin{definition}
\noindent We also recall the following notation. Let $(x_k)_k$ and $(y_k)_k$ be two sequences of real numbers. We let $y_k = O(x_k)$ if there exists $C_1 > 0$ such that $|y_k| \leq C_1 |x_k|$; $y_k = \Omega(x_k)$ if there exists $C_2 > 0$ such that $|y_k| \geq C_2 |x_k|$;
and $y_k = \Theta(x_k)$ if $y_k = O(x_k)$ and $y_k = \Omega(x_k)$.
%\end{definition}

% \noindent From Chebyshev's inequality, we immediately get:
% \begin{lemma}[Rate of convergence: Upper bound]
% Let $(X_k)_k$ be a sequence of unbiased estimators for a parameter $x$. 
% If $\var[X_k] = O(n_k^{-2\beta})$ for some $\beta > 0$, 
% then $|X_k - x| = O_p(n_k^{-\beta})$.
% \label{lem01}
% \end{lemma}
% \begin{proof}
% Since $\var[X_k] = O(n_k^{-2\beta})$, we can choose $M_\delta > 0$ for every $\delta > 0$ such that
% \[\sup_{k} \frac{n_k^{2\beta} \var[X_k]}{M^2_\delta} < \delta.\]
% Applying Chebyshev's inequality, we have
% \[\sup_{k} \P{n_k^\beta|X_k - x| \geq M_\delta} \leq \sup_{k} \frac{n_k^{2\beta} \var[X_k]}{M^2_\delta} < \delta.\]
% \end{proof}

% \noindent For the other direction:
% \begin{lemma}[Rate of convergence: Lower bound]
% \label{lemma:not-consistent}
% Let $(X_k)_k$ be a sequence of unbiased estimators 
% for a parameter $x$, such that $X_k \sim {\cal N}(x,\sigma_k^2)$. 
% For $\beta > 0$, if 
% $$\limsup_k \frac{\sigma_k^2}{n_k^{-2\beta}} = +\infty,$$ 
% then for all $M > 0$
% $$\limsup_k \P{n_k^\beta |X_k - x| > M} = 1,$$
% that is, $(X_k)_k$ is not $(n_k^{\beta})$-consistent.
% \end{lemma}
% \begin{proof}
% Note that $\sigma_k^{-1}(X_k - x) \sim {\cal N}(0,1)$.
% \end{proof}

\paragraph{Growth}
Our asymptotic results depend on how fast the tree grows. We first provide some intuition through
a toy example.
\begin{example}[Star tree: A first phase transition]
	\label{ex:firstphase}
	Let $\tree_k$ be a star tree with $n_k$ leaf edges of length
	$T_k$ emanating from the root. By symmetry, $\one$ is
	an eigenvector of $\mathbf{\Sigma}$ with eigenvalue
	$\lambda_k = \gamma[1+(n_k-1)e^{-2\alpha T_k}]$. Hence, $\one$ is also an
	eigenvector of $\mathbf{\Sigma}^{-1}$ with eigenvalue
	$\lambda_k^{-1}$ and
	$ \one'\mathbf{\Sigma}^{-1}\one = n_k\lambda_{n_k}^{-1},$
	so that $\mle^{(k)} = \overline{Y}$ and
	\begin{equation}\label{eq:varmlestar}
	\var[\mle^{(k)}]  = \frac{\lambda_k}{n_k}  % = \frac{1}{{n_k}^2} \cdot n_k \cdot \lambda_k 
	=  \gamma\left[e^{-2\alpha T_k} + \frac{1 - e^{-2\alpha T_k}}{n_k}\right].
	\end{equation}
	If both $n_k$ and $T_k \to +\infty$, then
	$\var[\mle^{(k)}] \to 0$ and 
	the MLE (and $\overline{Y}$)
	is consistent for $\mu$. Furthermore, if 
	$$\liminf_k \frac{2\alpha T_k}{\log n_k} > 1,$$
	then 
	$$ n_k \var[\mle^{(k)}] \leq \gamma[n_k e^{-2\alpha T_k} + 1] = O(1)$$
	% = \gamma \exp\left(\log n_k \left(1 -\frac{2\alpha T_k}{\log n_k}\right)\right) +\gamma
	and the MLE is $\sqrt{n_k}$-consistent
	(by an application of Chebyshev's inequality). % by Lemma~\ref{lem01}.
	On the other hand, if
	$$ \liminf_k \frac{2\alpha T_k}{\log n_k} < 1, $$
	then 
	$$ n_k \var[\mle^{(k)}] \geq \gamma[n_k e^{-2\alpha T_k}],$$
	% = \gamma \exp\left(\log n_k \left(1 -\frac{2\alpha T_k}{\log n_k}\right)\right),
	which goes to $+\infty$ along a subsequence, 
	and the MLE is {\em not} $\sqrt{n_k}$-consistent
	(using that $\mle$ is unbiased and normally distributed). %by Lemma~\ref{lemma:not-consistent}.
\end{example}
To study more general trees,
we use several standard notions of growth,
which play an important
role in random walks, percolation and ancestral
state reconstruction on trees (see e.g.~\cite{peres1999climb}).
%Fix a tree sequence $\treeseq = (\tree_k)_k$ with heights $(T_k)_k$ 
%and numbers of tips $(n_k)_k$.
\begin{definition}[Growth]
\label{def:growth}
The {\em lower growth} and 
{\em upper growth} of a tree sequence $\treeseq$
are defined respectively as
$$
\lgr = \liminf_k \frac{\log n_k}{T_k},
\quad\mbox{and}\quad
\ugr = \limsup_k \frac{\log n_k}{T_k}.
$$
In case of equality we define the growth
$\gr = \lgr = \ugr$.
(Note that our definition differs
slightly from~\cite{peres1999climb} in that
we consider the ``exponential rate'' of growth.)
\end{definition}
\noindent That is, for all $\epsilon > 0$, eventually
$
e^{(\lgr - \epsilon)T_k} \leq
n_k \leq e^{(\ugr + \epsilon)T_k},
$
and along appropriately chosen subsequences
$
n_{k_j} \geq e^{(\ugr - \epsilon)T_{k_j}}
$
%\quad\mbox{and}\quad
and
$
n_{k'_{j}} \leq e^{(\lgr + \epsilon)T_{k'_j}}.
$

We also need a stronger notion of growth.
For a tree $\tree$, thinking of the branches 
of $\tree$ as a continuum of {\em points},
a cutset $\pi$ is a set of points of $\tree$ such that all
paths from the root to a leaf must cross $\pi$.
Let $\Pi^k$ be the set of cutsets of $\tree_k$.
\begin{definition}[Branching number]
The {\em branching number} of $\treeseq$ 
is defined as
\[
\br = \sup \left \{ \Lambda \geq 0\ :\ 
\inf_{k, \pi \in \Pi^k} 
\sum_{x \in \pi}{ e^{- \Lambda \delta_k(\rho,x)}  > 0 } \right \},
\]
where $\delta_k(\rho,x)$ is the length of the path from the root 
to $x$ in $\tree_k$.
%(Again, unlike~\cite{peres1999climb}, 
%we consider the exponential rate of branching.)
\end{definition}
\noindent Because the leaf set $\leaves_k$ forms a cutset, it holds that
$$
\br \leq \lgr \leq \ugr.
$$
Unlike the growth, the branching number takes into
account aspects of the ``shape'' of the tree.
\begin{example}[Star tree sequence, continued]
\label{ex:first-transition-continued}
Consider again the setup of Example~\ref{ex:firstphase}.
The infimum
$$
\inf_{\pi \in \Pi^k} 
\sum_{x \in \pi} e^{- \Lambda \delta_k(\rho,x)},
$$
is achieved by taking $\pi = \leaves_k$ for every $k$. Hence
$\br = \lgr$. We showed in Example~\ref{ex:firstphase}
that the MLE of $\mu$ given $\alpha$ is
$\sqrt{n_k}$-consistent if
$\ugr < 2\alpha $, but not $\sqrt{n_k}$-consistent if
$\ugr > 2\alpha $.  
\end{example}

Finally, we will need a notion of uniform growth.
\begin{definition}[Uniform growth]
Let $\treeseq = (\tree_k)_k$ be a tree sequence.
For any point $x$ in $\tree_k$, let $n_k(x)$
be the number of leaves below $x$ and
let $T_k(x)$ be the distance from $x$ to the leaves.
Then the {\em uniform growth} of $\treeseq$
is defined as
$$
\unigr 
= \lim_{M\to+\infty} 
\sup_{k, x \in \tree_k}\ \frac{\log n_k(x)}{T_k(x)\lor M}.
$$
(The purpose of the $M$ in the denominator
is to alleviate boundary effects.)
\end{definition}

\subsection{Statement of results}
\label{sec:statementresults}
We can now state our main results.

\paragraph{Results concerning the mean $\mu$}
We first give a characterization of the consistency of 
the MLE of $\mu$. In words, the MLE sequence
is consistent if, in the limit, we can find arbitrarily many
descendants, arbitrarily far away from the
leaves. 
% In particular this criterion implies that under Assumption~\ref{assumption:unbounded}
% the MLE of $\mu$ is always consistent on nested and growing sequences.
This theorem is proved in Section~\ref{sec:consistency-mle},
along with a related result involving the
branching number.  
%\cnote{and along with other things?}
\begin{thm}[Consistency of $\mle$]
\label{thm:criterion-consistency}
Let $(\tree_k)_k$ be a sequence of trees satisfying 
Assumption~\ref{assumption:unbounded}. 
Let $(\mle^{(k)})_k$ be the corresponding sequence
of MLEs of $\mu$ given $\alpha$. Denote by $\tilde{\pi}_t^{k}$ the 
cutset of $\tree_k$ {\em at time $t$ away from the leaves} 
and let $T_k$ be the height
of $\tree_k$. Then $(\mle^{(k)})_k$ is consistent
for $\mu$ if and only if for all $s \in (0,+\infty)$
\begin{equation}\label{eq:consistency-condition}
\liminf_k \left|\tilde{\pi}^k_{s}\right| = +\infty.
\end{equation}
\end{thm}

We further obtain bounds on the
variance of the MLE to characterize the rate of convergence
of the MLE. In particular we give conditions for
$\sqrt{n_k}$-consistency. We show that the latter
undergoes a phase transition, generalizing Example~\ref{ex:first-transition-continued}. 
When the upper growth is above $2\alpha$,
we show that the MLE of $\mu$ cannot be 
$\sqrt{n_k}$-consistent. %(Theorem~\ref{thm:rate-supercritical}). 
If further the branching number is above $2\alpha$, we give tight bounds on the
convergence rate of the MLE. Roughly we show that,
in the latter case, the variance behaves like $n_k^{2\alpha/\gr}$.
Or perhaps a more accurate way to put it is that the ``effective
number of samples'' $\neff_k$ is $e^{2\alpha T_k}$,
in the sense that $\var_{\tree_k}[\mle^{(k)}] = \Theta((\neff_k)^{-1})$.

\begin{thm}[Loss of $\sqrt{n_k}$-consistency for $\mle$: Supercritical regime]
\label{thm:rate-supercritical}
Let $(\tree_k)_k$ be a tree sequence.
If $\ugr > 2\alpha$, then 
for all $\epsilon > 0$ there is a subsequence $(k_j)_j$
along which
\begin{equation}\label{eq:non-sqrt-supercritical}
\var_{\tree_{k_j}}[\mle^{(k_j)}] 
\geq \gamma n_{k_j}^{-2\alpha/(\ugr - \epsilon)}.
\end{equation}
In particular
$(\mle^{(k)})_k$ is not
$\sqrt{n_k}$-consistent.
If, further,
\begin{enumerate}
\item $\br > 2\alpha$: then
$$
\var_{\tree_k}[\mle^{(k)}] = \Theta\left(e^{-2\alpha T_k}\right).
$$
Moreover in terms of $n_k$, for all $\epsilon > 0$, there are constants
$0 < C', C < +\infty$ such that
\begin{equation}\label{eq:br-above-upper}
C' n_k^{-2\alpha/(\lgr - \epsilon)}
\leq 
\var_{\tree_k}[\mle^{(k)}]
\leq 
C n_k^{-2\alpha/(\ugr + \epsilon)},
\end{equation}
and, in addition to~\eqref{eq:non-sqrt-supercritical}, 
\begin{equation*}
%\label{eq:subseq-lower}
\exists\ \text{subsequence}\ (k'_j)_j,\ \text{s.t.}\ \var_{\tree_{k'_j}}[\mle^{(k'_j)}]
\leq \gamma n_{k'_j}^{-2\alpha/(\lgr + \epsilon)}.
\end{equation*}

\item $\br < 2\alpha$: then, for all $\epsilon > 0$, there are constants
$0 < C', C < +\infty$ such that
\begin{equation}\label{eq:br-below-upper}
C' n_k^{-2\alpha/(\lgr - \epsilon)}
\leq 
\var_{\tree_k}[\mle^{(k)}]
\leq 
C n_k^{-(\br - \epsilon)/(\ugr + \epsilon)},
\end{equation}
where the lower bound  in~\eqref{eq:br-below-upper} above
holds provided $\lgr > 0$, and

\begin{equation*}
%\label{eq:subseq-lower-2}
\exists\ \text{subsequence}\ (k'_j)_j,\ \text{s.t.}\ \var_{\tree_{k'_j}}[\mle^{(k'_j)}]
\leq \gamma n_{k'_j}^{-(\br - \epsilon)/(\lgr + \epsilon)}.
\end{equation*}

\end{enumerate}
\end{thm}

\noindent
The following example shows that, when $\br < 2\alpha$,
the upper bound in \eqref{eq:br-below-upper} may not be achieved,
but cannot be improved in general.
\begin{example}[Two-level tree]
\label{ex:twolevel}
Let $(\tree_k)_k$ be a tree sequence 
% spherically symmetric trees 
% fixit: remove this here, but say in ex:spherical that we generalize ex:twolevel
% bas defined in Example~\ref{ex:spherical},
with two levels of nodes below the root: % $H=2$
$D^{(k)}_0= e^{ \Lambda_0 \tau^{(k)}_0}$ nodes are attached to the root 
by edges of length $\tau^{(k)}_0 = \sigma T_k$, 
for some arbitrary choice of tree height $T_k\to\infty$ 
and $0<\sigma<1$.
Each of these $D^{(k)}_0$ nodes has itself
$D^{(k)}_1 = e^{ \Lambda_1 \tau^{(k)}_1}$ children 
along edges of length $\tau^{(k)}_1 = (1-\sigma) T_k$, 
and these form the leaves of $\tree_k$.

\begin{prop}\label{prop:example2level}
For $0<\Lambda_0 < \Lambda_1$ and $\treeseq=(\tree_k)_k$ described above, 
we have that $\br = \Lambda_0$, 
$\gr = \sigma \Lambda_0 + (1-\sigma) \Lambda_1$, and
\begin{eqnarray} 
\var_{\tree_k}[\mle^{(k)}] &=& \gamma  e^{-2\alpha T_k}  + 
\gamma (1 - e^{-2\alpha \sigma T_k}) \, e^{-(\sigma\Lambda_0 + (1-\sigma)2\alpha) T_k} \nonumber\\
&&\qquad\qquad\qquad\qquad\qquad + \gamma (1 - e^{-2\alpha (1-\sigma) T_k}) \, e^{-\gr T_k} .\label{eq:variance-2-level}
\end{eqnarray}
\end{prop}

\noindent This proposition is proved in Section \ref{sec:convergence-rate-mle}. 
It implies that if $ 2\alpha \leq \Lambda_0 = \br$,
the dominant term in the variance
is $\gamma e^{-2\alpha T_k} = \gamma n_k^{-2\alpha /\gr}$, as predicted
by \eqref{eq:br-above-upper} in Theorem~\ref{thm:rate-supercritical}.
If instead $2\alpha \geq \Lambda_1 $,
the dominant term in the variance is 
$\gamma e^{- \gr T_k} = \gamma n_k^{-1},$ 
and we have $\sqrt{n_k}$-consistency. 
In the intermediate case when $\Lambda_0 < 2\alpha < \Lambda_1$,
the dominant term in the variance is 
$ \gamma e^{- (\sigma \Lambda_0 + (1-\sigma) 2\alpha) T_k}
= \gamma n_k^{-\sigma (\br/\gr) - (1-\sigma) (2\alpha/\gr)}.$ 
Therefore, depending on the value of $\sigma$,
we can get the full range of exponent values 
between $-2\alpha/\gr$ and $-\br/\gr$, as given in~\eqref{eq:br-below-upper}.
\end{example}

\noindent In the other direction when $\ugr < 2\alpha$, 
the picture is somewhat murkier.
For example, by taking % $\Lambda_1$ large enough and
$\sigma$ close enough to $1$ in Example~\ref{ex:twolevel},
it is possible to have $\gr < 2\alpha$, yet
not $\sqrt{n_k}$-consistency. 
The issue in Example~\ref{ex:twolevel} is the inhomogeneous growth rate.
However, under extra regularity conditions,
$\sqrt{n_k}$-consistency can be established.
In words, the growth of the tree must be
sufficiently homogeneous. 
In Theorem~\ref{thm:subcritical} below,
we consider imposing the extra condition $\br = \gr$, which 
does not hold in Example~\ref{ex:twolevel}.
%fixit
%(In fact if $\ugr = 0$ we may not have
%consistency, as Example~\ref{ex:spherical-bad}
%shows.) 
\begin{thm}[Convergence rate of $\mle^{(k)}$: Subcritical regime]
\label{thm:subcritical}
Let $(\tree_k)_k$ be a tree sequence with
$\ugr < 2\alpha$. Then
$$
\var_{\tree_k}[\mle^{(k)}] = \Omega\left(n_k^{-1}\right).
$$
Further if:
\begin{enumerate}
\item %{\em [Equality condition]} 
$\br = \ugr > 0$ then, for all $\epsilon > 0$,
$
\var_{\tree_k}[\mle^{(k)}] = O\left(n_k^{-(1-\epsilon)}\right).
$

\item %{\em [Bounded uniform growth]} 
$\unigr < 2\alpha$ then
$
\var_{\tree_k}[\mle^{(k)}] = O\left(n_k^{-1}\right).
$

\end{enumerate}
\end{thm}

Theorems~\ref{thm:rate-supercritical} and~\ref{thm:subcritical} are proved in Section~\ref{sec:convergence-rate-mle}. All our results
on the estimation of $\mu$ leverage a 
useful characterization of the variance 
of linear estimators in terms
of electrical networks. An analogous
characterization is used in ancestral state reconstruction~\cite{peres1999climb}. Note that our results are not as
clean as those obtained for ancestral
state reconstruction. As Example~\ref{ex:twolevel}
showed, estimation of $\mu$ is somewhat sensitive to
the ``homogeneity'' of the growth.
In Section~\ref{sec:sensitivity-to-alpha},
we show that assuming $\alpha$ is known 
is inconsequential, provided a good estimate of $\alpha$
is available. Such an estimate is discussed next.

\paragraph{Results concerning the parameters
$\alpha$ and $\gamma$}
Our main result for $\alpha$ and $\gamma$ is a
$\sqrt{n_k}$-consistent estimator under the following assumption:
there are two separate ``bands'' of node ages, each containing
a number of internal nodes growing linearly 
with the number of leaves. 
\begin{assumption}[Linear-sized bands]
\label{assump:alpha}
Define $n_k(c,c')$ as the number of nodes in $\tree_k$ 
of age (height from the leaves) in $(c,c')$.
Assume that there are constants $\beta >0$ and
$0 < c_1 < c_1' < c_2 < c_2' < \infty$ such that
$n_k(c_i,c_i') \geq \beta n_k, i =1,2,$
for all $k$ large enough.
\end{assumption}

As shown in Corollary \ref{cor:yule_alpha_gamma}, this assumption holds for the Yule process, a speciation model frequently used in practice.

\begin{thm}[Estimating $\alpha$ and $\gamma$: $\sqrt{n_k}$-consistency]
\label{thm:main-alpha-rate}
Let $(\tree_k)_k$ be a sequence of ultrametric trees
satisfying Assumptions~\ref{assumption:unbounded}
and~\ref{assump:alpha}.
Then there is an estimator $(\hat\alpha_k,\hat\gamma_k)_k$ 
of $(\alpha,\gamma)$ such that
$|\hat\alpha_k -\alpha| = O_p(n_k^{-1/2})$ 
and $|\hat\gamma_k -\gamma| = O_p(n_k^{-1/2})$.
\end{thm}

The proof, found in Section~\ref{sec:alpha-rate}, is based
on the common notion of contrasts. Assumption~\ref{assump:alpha}
ensures the existence of an appropriate set of such contrasts.
The key point is that this extra assumption can be satisfied
no matter what the growth and branching number are,
indicating that the estimation of $\alpha$ and $\gamma$
is unaffected by the growth of the tree unlike $\mu$.
Intuitively, $\mu$ is a more ``global'' parameter.

\subsection{Special cases}
\label{sec:special-cases-mu}

We apply here the results stated in Section \ref{sec:statementresults} 
to a number of scenarios.
The tree of life naturally gives rise to two types of tree sequences.
If one imagines sampling an increasing number of contemporary species, one obtains a nested sequence, defined as follows.
\begin{definition}[Nested sequence]
A sequence of trees $\treeseq=(\tree_k)_{k}$ is {\em nested}
if, for all $k$, $n_k = k$ and $\tree_k$ restricted to 
%$[k-1]$ 
the first $k-1$ species
is identical to $\tree_{k-1}$ as an ultrametric.
\end{definition}
\noindent An example of nested trees is given by a caterpillar sequence.
\begin{example}[Caterpillar sequence]
	\label{ex:caterpillar}
	Let $(t_k)_k$ be a sequence of nonnegative numbers such that
	$\limsup_k t_k = +\infty$.
	Let $\tree_1$ be a one-leaf star with height $T_1 = t_1$. For $k > 1$, let $\tree_k$ be the caterpillar-like tree
	obtained by adding a leaf edge with leaf $k$ to $\tree_{k-1}$ 
	at height $t_k$ on the path between leaf $1$ and the root
	of $\tree_{k-1}$, if $t_k \leq T_{k-1}$. 
	If instead
	$t_k > T_{k-1}$, create a new root at height $t_k$ 
	with an edge attached to
	the root of $\tree_{k-1}$ and an edge attached to $k$ (see Figure \ref{fig:caterpillar}).
	%We use a caterpillar tree later to provide an example where
	%the sample mean $\overline{Y}$ is not consistent (Example~\ref{ex:caterpillar-bad}),
	%thus performs significantly worse than the MLE $\mle$.
\end{example}

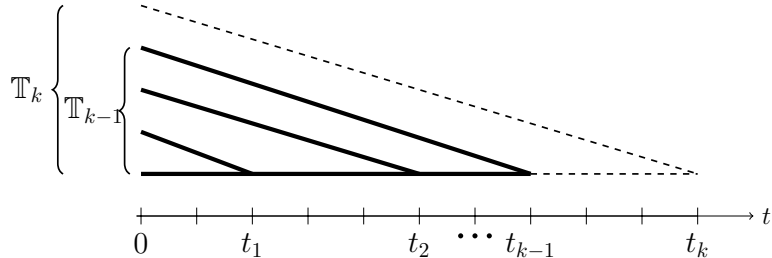
\begin{figure}[h]
\begin{center}
 \scalebox{0.8}{%
\begin{tikzpicture}[domain=0:1,xscale=3.7,yscale=0.7] %,style={font=\footnotesize}]

\draw[very thin,color=black] (-0.02,-0.2) grid[xstep=0.25] (2.5,0.2); % grid on X axis
\draw[->] (-0.02,0) -- (2.75,0) node[right] {$t$}; % X axis

\draw (0,-0.2) node[below] {\large $0$};
\draw (0.5,-0.2) node[below] {\large $t_1$};
\draw (1.25,-0.2) node[below] {\large $t_2$};
\draw (1.5,-0.2) node[below] {\huge $...$};
\draw (1.75,-0.2) node[below] {\large $t_{k-1}$};
\draw (2.5,-0.2) node[below] {\large $t_{k}$};

\draw[line width=2pt] plot coordinates {(0,1) (1.75, 1)};
\draw[line width=2pt] plot coordinates {(0,2) (0.5, 1)};
\draw[line width=2pt] plot coordinates {(0,3) (1.25, 1)};
\draw[line width=2pt] plot coordinates {(0,4) (1.75, 1)};
\draw[dashed, thick] plot coordinates {(0,5) (2.5, 1)};
\draw[dashed, thick] plot coordinates {(1.75,1) (2.5, 1)};

\draw[decorate,decoration={brace,amplitude=5pt},thick] 
    (-0.05,1)  -- (-0.05,4) node[black, midway, xshift=-0.6cm]{\large $\tree_{k-1}$};
\draw[decorate,decoration={brace,amplitude=5pt},thick] 
    (-0.35,1)  -- (-0.35,5) node[black, midway, xshift=-0.6cm]{\large $\tree_{k}$};

\end{tikzpicture}
}
\end{center}
\caption{Example of a sequence of nested caterpillar trees.}
\label{fig:caterpillar}
\end{figure}

\begin{cor}[Nested sequence: consistency of $\mle$]
\label{cor:consistency-nested}
%Let $\treeseq$ be a nested sequence satisfying Assumption~\ref{assumption:unbounded}.
Let $\treeseq$ be a nested sequence such that the height $T_k$ goes to infinity. 
Then $\treeseq$ satisfies Assumption~\ref{assumption:unbounded} and
the MLE for $\mu$ is consistent on $\treeseq$.
\end{cor}
\begin{proof}
Let $k_j$ be the subsequence such that 
$T_{k_{j+1}} > T_{k_j}$ 
for every $j$ 
and $T_i = T_{k_j}$ 
for all $i = k_j + 1, \ldots, k_{j+1}-1$.
Then, for all
$s \in (0,+\infty)$, as $k$ goes to $+\infty$
$\pi^k_s$ eventually contains
all leaves $k_j$ such that $T_{k_j} \geq s$.
Since $T_k \to +\infty$, 
%by Assumption~\ref{assumption:unbounded},
the result follows.  
\end{proof}
%\noindent For nested sequences, it was shown in~\cite{HoAne13}
%that Assumption~\ref{assumption:unbounded} is necessary,
%as on bounded-height tree sequences the MLE
%of $\mu$ is not consistent. 

If one is modeling the growth
of the tree of life in time, instead of modeling increased sampling of 
contemporary species, one obtains a growing
sequence as follows. Let $\tree_0$ be a rooted
infinite tree of bounded degree, with branch lengths and no leaves.
Think of the branches of $\tree_0$ as a continuum
of {\em points} whose distance from the endpoints grows
linearly. Then, for $t \geq 0$, 
we define $\mathcal{B}_t(\tree_0)$ as the tree made of
the set of points of $\tree_0$ at distance at most $t$
from the root.
\begin{definition}[Growing sequence]
A sequence of trees $(\tree_k)_k$ is a 
{\em growing sequence} of trees if there is 
an infinite tree $\tree_0$ as above and
an increasing sequence of non-negative reals
$(t_k)_k$ such that $\tree_k$ is isomorphic
to $\mathcal{B}_{t_k}(\tree_0)$ as an ultrametric.
\end{definition}

\begin{cor}[Growing sequence: consistency of $\mle$]
%Let $(\tree_k)_k$ be a growing sequence satisfying Assumption~\ref{assumption:unbounded}.
Let $(\tree_k)_k$ be a growing sequence such that the height $T_k = t_k$ goes to infinity.
Then $\treeseq$ satisfies Assumption~\ref{assumption:unbounded} and
the MLE for $\mu$ is consistent on $(\tree_k)_k$.
\end{cor}
\begin{proof}
Fix $s \in (0,+\infty)$. 
For $L =1,2,\ldots$, 
let $k'_L$ be the smallest
$k$ such that $n_k \geq L$
and let $k''_L$ be the smallest 
$k > k'_L$ such that 
$T_{k} \geq T_{k'_L} + s$.
Then, for all $k \geq k''_L$,
$|\pi^k_s|\geq L$. Letting $L$
go to $+\infty$ gives the result.
\end{proof}

\begin{example}[Yule sequence]
\label{ex:yule}
Let $\tree_0$ be a tree generated by a pure-birth (Yule)
process with rate $\lambda > 0$: starting with one lineage,
each current lineage splits independently after an exponential
time with mean $\lambda^{-1}$ 
(see e.g.~\cite{semple2003phylogenetics}).
For any (possibly random) sequence
of increasing non-negative reals $(t_k)_k$
with $t_k \to +\infty$, 
$\mathcal{B}_{t_k}(\tree_0)$ (that is, 
$\tree_0$ run up to time
$t_k$), forms a growing sequence.
\end{example}
The following result is proved in Section~\ref{section:proofs-special}.
\begin{cor}[Yule model: consistency of $\mle$]
\label{cor:mu-yule}
Let $(\tree_k)_k$ be a Yule sequence with rate
$0 < \lambda < +\infty$. Then, with probability
$1$ (on the generation of $\tree_0$),
\begin{enumerate}
\item $(\mle^{(k)})_k$ is consistent.

\item If $\lambda < 2 \alpha$, $(\mle^{(k)})_k$ is $\sqrt{n_k}$-consistent.

\item If $\lambda > 2 \alpha$, $(\mle^{(k)})_k$ is 
{\em not} $\sqrt{n_k}$-consistent and for all
$\epsilon > 0$ there is $0 < C', C < +\infty$
such that
$$
C' n_k^{-2\alpha\lambda^{-1}-\epsilon}
\leq
\var_{\tree_k}[\mle^{(k)}]
\leq
C n_k^{-2\alpha\lambda^{-1}+\epsilon}.
$$
\end{enumerate}
\label{cor:yule}
\end{cor}

We also apply the estimators $\hat\alpha$ and $\hat\gamma$ to the Yule model.
For simplicity, we take the sequence of 
times at which new speciation events occur (although this assumption is not crucial).
% Let $\tree_0$ be an infinite Yule tree with
% rate $0 < \lambda < +\infty$. 
For $k \geq 1$, let $t_k$ be the first time 
at which $\tree_0$
has $k+1$ lineages. Then $n_k = k$ for all $k$
and $t_k \to +\infty$ so that
Assumption~\ref{assumption:unbounded} is
satisfied. The following result is proved in Section~\ref{section:proofs-special}.
\begin{cor}[Yule model: estimation of $\alpha$ and $\gamma$]
\label{cor:yule_alpha_gamma}
Let $(\tree_k)_k$ be a Yule sequence with $n_k=k$ as above. Then
Assumption~\ref{assump:alpha} is satisfied 
asymptotically, 
and hence
$|\hat\alpha_k -\alpha| = O_p(n_k^{-1/2})$ 
and $|\hat\gamma_k -\gamma| = O_p(n_k^{-1/2})$.
\end{cor}

\section{Application to experimental design for trait evolution studies}
\label{sec:app}

Thanks to recent developments in technology, scientists have reconstructed several large phylogenetic trees with thousands of species such as 
trees containing 4507 mammal species \cite{binindaEmonds-etal07} 
and 9993 bird species \cite{jetz12}. 
However, researchers may not be able to collect trait data 
from all species, due to limited resources and funding. Thus, many studies are only based on a subset of species in the available tree. For example, to study the evolution of body size in mammals, Cooper and Purvis \cite{cooperPurvis10} 
used 3473 of the 4507 species in their tree, and Venditti et al. \cite{vendittiMeadePagel11} incorporated 3185 species in their analysis. 
When considering extra data collection, an important question arises: can additional species increase the precision of our estimates? 
Our theoretical results help answering this question for the OU tree model:
\begin{enumerate}	
	\item If $\hat \lambda \ll 2 \hat \alpha$, additional species tend to be very informative for estimating $\mu$ (Corollary \ref{cor:mu-yule}).
	\item If $\hat \lambda \gg 2 \hat \alpha$, additional species that do not increase tree height tend to be non-informative for estimating $\mu$ (Corollary \ref{cor:mu-yule}). 
	\item When $\hat \lambda$ is around $2 \hat \alpha$, it is not clear whether additional species are informative for estimating $\mu$.
	\item Additional species tend to be informative for estimating $\alpha$ and $\gamma$ (Corollary \ref{cor:yule_alpha_gamma}). 
\end{enumerate}

%We illustrate our suggestions by an example using the data from the study of body size evolution in mammal \cite{vendittiMeadePagel11}, which only consists the body size of 3185 species (out of the available 4507-species mammal tree).

\paragraph*{Example:~~}  
In \cite{vendittiMeadePagel11}, body size evolution was studied using 
3185 mammal species. Would it be worth the effort to collect data for the 
remaining 1322 species in the tree, to increase 
the precision of estimating $\mu$? 
To answer this question about sampling utility, we first need to estimate the speciation rate $\lambda$ and the selection strength $\alpha$.  The 4507-species mammal tree was rescaled to have height $1$ and its speciation rate was estimated to be $11.83$ using
maximum likelihood (\texttt{yule} function in the R package \texttt{ape} \cite{ape}). 
We also estimated $\hat\alpha=0.01$ using maximum likelihood (\texttt{phylolm} function in the R package \texttt{phylolm} \cite{HoAne14sb}).  
Note that the tree formed by the 3185 species has the same height as the full tree with all 4507 species. Since $\hat \lambda \approx 11.83 \gg 0.02 \approx 2 \hat \alpha$, additional species tend to be non-informative 
and our recommendation is to stop data collection. 
Our conclusion is consistent with simulations in \cite{HoAne13,ho2014intrinsic}, which 
showed that additional species are non-informative for estimating $\mu$ 
if they do not increase tree height, 
when $\alpha$ is low. 
Our recommendation here specifies the critical value of $\alpha$ below which
additional sampling is of little utility.

To further demonstrate the 
relationship between sampling utility and $\alpha$ (or $\lambda$)
at fixed tree height, 
we simulated data according to the OU model along the 4507-species mammal tree with $\mu = 0$, $\gamma = 1$, and several values of $\alpha$ ranging from $0.01$ to $300$. For every set of parameters, we simulated 2000 data sets using
the \texttt{rTrait} function (R package \texttt{phylolm}). Then, $\mle$ 
was computed for each data set using the \texttt{phylolm} function. 
The sample variance of $\mle$ (Figure \ref{fig:fig3}) was found to be
about $e^{-2 \alpha}$ when $2 \alpha \ll \hat \lambda$, and 
about $1/n = 1/4507$ when $2 \alpha \gg \hat \lambda$. 

\begin{figure}[!h]
\begin{center}
\includegraphics[scale=.5]{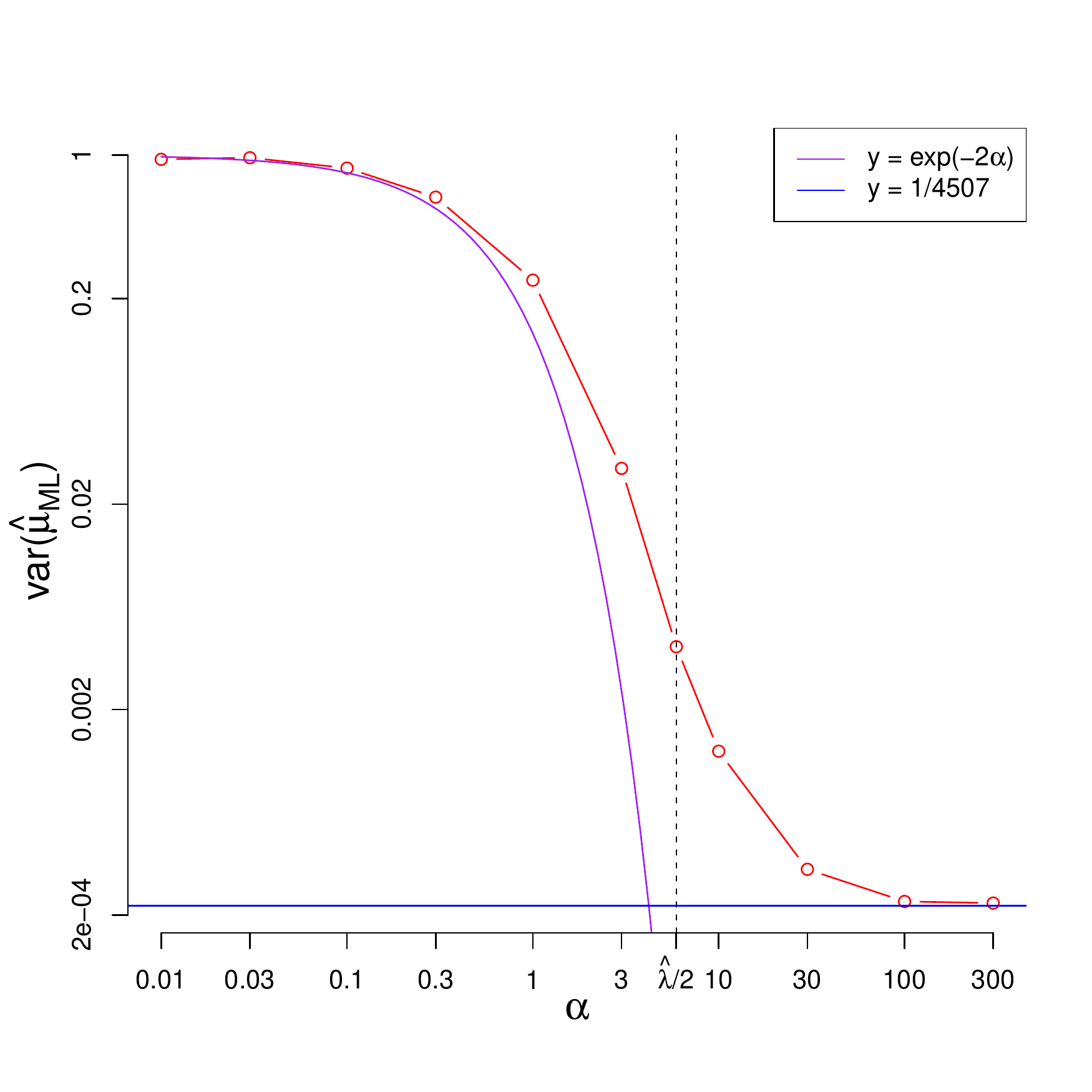}
\caption{Sample variance of $\mle$ (red points) as a function of $\alpha$,
from the simulation on the mammal tree. When $\alpha$ is small, $\text{var}(\mle) \approx e^{-2 \alpha}$ (purple line). When $\alpha$ is large, $\text{var}(\mle) \approx 1/4507$ (blue line).}
\label{fig:fig3}
\end{center}
\end{figure}

To illustrate the relationship between sampling utility and $\alpha$ (or $\lambda$)
when the tree height varies, 
we simulated 400 trees under the Yule process 
using the \texttt{sim.bdtree} function (R package \texttt{geiger} \cite{geiger}).
We used speciation rate $\lambda = 11.83$, which was the maximum
likelihood estimate from the mammal tree. 
The tree height was varied from $0.05$ to $1$ and we simulated $20$ trees for each tree height. We calculated $\text{var}(\mle)$ corresponding to three fixed values of $\alpha$ $(0.1,\lambda/2,30)$ using \eqref{eq:varmuhat} 
and the \texttt{three.point.compute} function (R package \texttt{phylolm}). The results showed that (Figure \ref{fig:fig4}) when $2 \alpha \ll \lambda$, $e^{-2 \alpha T}$ approximates $\text{var}(\mle)$ better than $1/n$. On the other hand, when $2 \alpha \gg \lambda$, $1/n$ is a better approximation.
 
\begin{figure}[!h]
\begin{center}
\includegraphics[scale=.7]{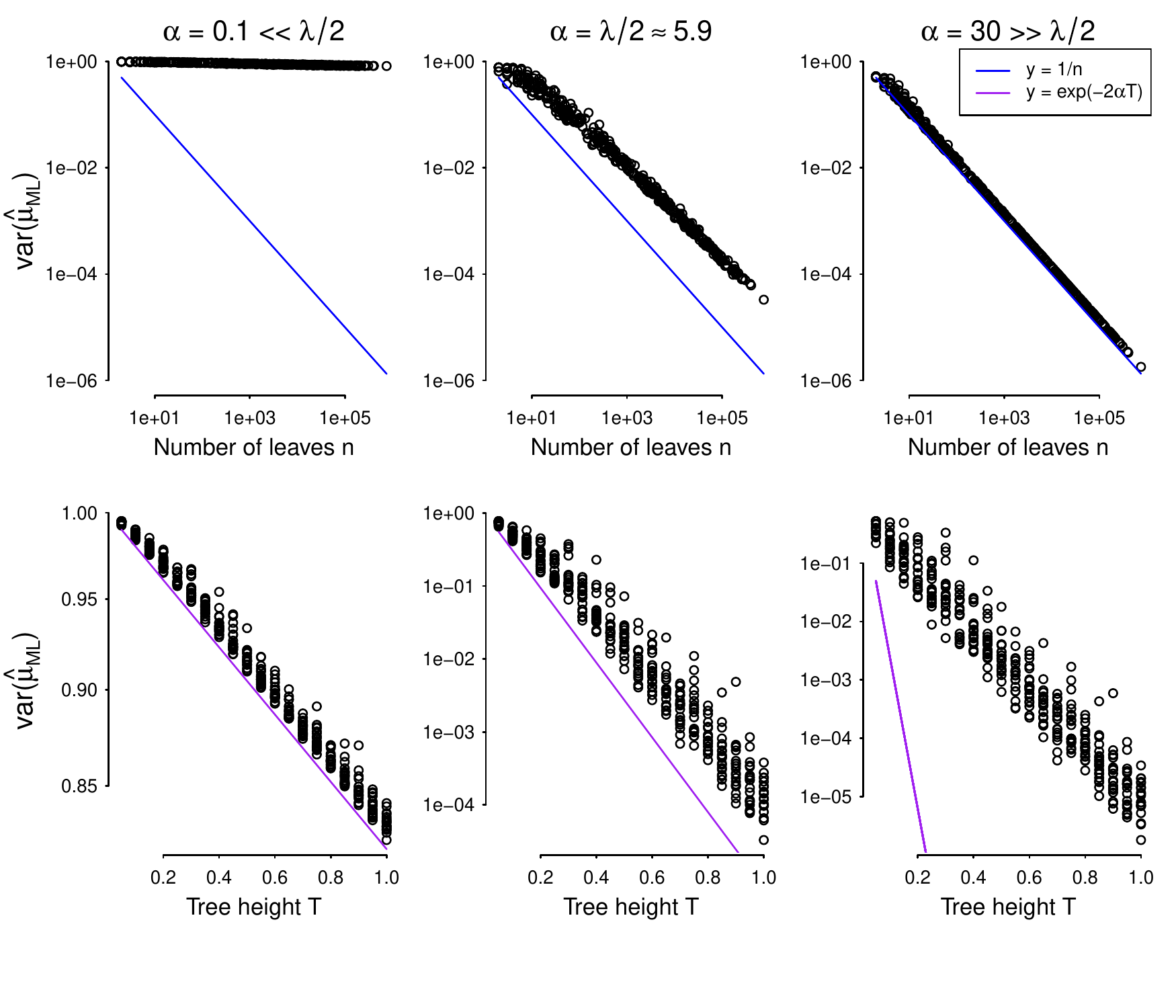}
\caption{Variance of $\mle$ on ramdom trees, simulated
under the Yule process. Top: $\text{var}(\mle)$ against number of leaves $n$.
Bottom: $\text{var}(\mle)$ against tree height $T$. 
The true value of $\alpha$ is either small (0.1, left), or $\lambda/2$ (5.9, middle), or large (30, right).}
\label{fig:fig4}
\end{center}
\end{figure}

Taken together, our results show that when $2 \alpha \ll \hat \lambda$, 
the variance of $\mle$ depends on the tree height, not the sample size. 
So, additional sampling that does not increase tree height is not recommended. 
On the other hand, when $2 \alpha \gg \hat \lambda$ the variance of $\mle$ 
is of order $1/n$, as if we had $n$ independent samples. 
In this case additional species are very informative, and additional sampling is 
recommended if affordable.

\section{Proofs of results for estimating $\mu$}
\label{sec:proofs}

We develop here necessary tools (Section \ref{sec:mu-tools}), then prove 
Theorem~\ref{thm:criterion-consistency} (Section \ref{sec:consistency-mle}),
Theorems~\ref{thm:rate-supercritical} and \ref{thm:subcritical} 
(Section \ref{sec:convergence-rate-mle}),
which assume that $\alpha$ is known.
Using arguments from the proofs, we also identify examples showing
that the sample mean $\overline Y$ can perform significantly worse
than $\mle$, and we show that Assumption~\ref{assumption:unbounded} 
is not sufficient in Theorem~\ref{thm:criterion-consistency} 
for the consistency of $\mle$. 
We prove an alternative sufficient condition based on the branching number
(Proposition~\ref{thm:consistency-branching} below).
In Section~\ref{section:proofs-special},
we prove Corollaries~\ref{cor:mu-yule} and~\ref{cor:yule_alpha_gamma}.
Finally, in Section~\ref{sec:sensitivity-to-alpha}
we discuss the sensitivity 
of the MLE to estimation errors on $\alpha$.

\subsection{Bounding the variance of the MLE}
\label{sec:mu-tools}

Fix an ultrametric species tree $\tree$
with leaf set $\leaves$, number of tips $n = |\leaves|$,
and root $\rho$. We also fix $\alpha > 0$.

\paragraph{A formula for the variance}
Let $\bftheta = (\theta_\ell)_{\ell\in\leaves}$,
with $\bftheta' \one  = 1$ and $\theta_\ell \in [0,1]$ for all $\ell$, 
and recall that
$
Y_\bftheta = \sum_{\ell \in \leaves} \theta_\ell Y_\ell
$
is an unbiased estimator of $\mu$.
By defining, for each branch $b$,
\begin{equation}\label{eq:flow-on-branch}
\theta_b = \sum_{\ell \in \mathscr{L}} \ind_{b \in p(\rho,\ell)}\theta_\ell,
\end{equation}
where $p(\rho,\ell)$ is the path from $\rho$ to $\ell$,
we naturally associate to the coefficients $\bftheta$
a flow on the edges of $\tree$, defined as follows.
\begin{definition}[Flow]
A {\em flow} $\bfeta$ is a mapping from the set of edges to the set of positive numbers such that, for every edge $b$, we have
$
\eta_b = \sum_{b' \in O_b}{\eta_{b'}}
$
where $O_b$ is the set of outgoing edges stemming from $b$
(with the edges oriented away from the root).
Define $\|\bfeta\| = \sum_{b \in O_\rho}{\eta_b}$.
We say that $\bfeta$ is a {\em unit flow} 
if $\|\bfeta\| = 1$.
We extend $\bfeta$ to vertices $v$ in $\tree$ by
defining $\eta_v$ as the flow on the edge entering $v$.
Similarly, for a point $x$ in $\tree$, we let $\eta_x$ 
be the flow on the corresponding edge or vertex.
\end{definition}
\noindent For every edge $b$ of $\tree$, 
we set
$
R_b = (1 - e^{-2 \alpha |b|}) e^{2 \alpha \delta(\rho,b)}
$
where $|b|$ is the length of $b$ 
and $\delta(\rho,b)$ is the length of the path from the root to $b$ (inclusive). 
\begin{prop}[Variance of $\mle$: Main formula]
\label{prop:key}
Let $\mathcal{F}$ be the set of unit flows from $\rho$ to $\leaves$. Let $E$ be the set of edges and $T$ be the height of tree $\tree$.
%\label{lemma:variance-formula} \label{lemma:network}
For any $\bftheta \in \mathcal{F}$, we have
\begin{equation}\label{eq:variance-formula}
\var[Y_\bftheta] = \gamma e^{- 2 \alpha T} \left( 1 + \sum_{b \in E}{R_b\theta^2_b} \right)
\end{equation}
so that
%\begin{eqnarray*}
$\var[\mle] = 
\inf_{\bftheta \in \mathcal{F}}\  
\gamma e^{- 2 \alpha T} ( 1 + 
\sum_{b \in E}{R_b\theta^2_b}).
$
%\end{eqnarray*}
\end{prop}

\noindent As detailed in~\cite{peres1999climb}, 
a species tree can be interpreted as an electrical network 
with resistance $R_b$ on edge $b$.
The minimum $\mathcal{R}_\tree$ of $\sum_{b \in E}{R_b\theta^2_b}$ over unit flows
(corresponding to the MLE) is known as 
the {\em effective 
resistance} of $\tree$, which can be interpreted in terms
of a random walk on the tree. See~\cite{peres1999climb}
for details.

\begin{proof}
The second part follows from the first, because $\mle$ is the best
unbiased linear estimator of $\mu$.
The proof of \eqref{eq:variance-formula} follows from a computation in~\cite[Lemma 5.1]{EvKePeSc:00}.
For every node $u$ of the tree, by a telescoping argument,
\begin{equation}\label{eq:telescope}
e^{2 \alpha \delta(\rho,u)} - 1 
= \sum_{b \in p(\rho,u)}{R_b}
\end{equation}
where $\delta(\rho,u)$ is the distance from $\rho$ to $u$, and $p(\rho,u)$ is the path from $\rho$ to $u$.
Denote by $v \wedge w$ the most recent common ancestor of $v$ and $w$. 
Then
\begin{eqnarray*}
\var[Y_\bftheta] &=& \gamma  \sum_{v,w \in \mathscr{L}}{\theta_v\theta_w \frac{e^{-2 \alpha T}}{e^{-2 \alpha \delta(\rho, v \wedge w)}}}
= \gamma e^{- 2 \alpha T} \sum_{v,w \in \mathscr{L}}{\theta_v\theta_w\bigg( 1 + \sum_{b \in p(\rho,v \wedge w)}{R_b} \bigg)}\\
&=& \gamma e^{- 2 \alpha T} \bigg( 1 + \sum_{b \in E}{R_b \sum_{v,w \in \mathscr{L}}{\ind_{b \in p(\rho,v \wedge w)}\theta_v\theta_w}} \bigg)\\
&=& \gamma e^{- 2 \alpha T} \left [ 1 + \sum_{b \in E}{R_b \Big( \sum_{v \in \mathscr{L}}{\ind_{b \in p(\rho,v)}\theta_v} \Big) \Big( \sum_{w \in \mathscr{L}}{\ind_{b \in p(\rho,w)}\theta_w} \Big)} \right ]\\
&=& \gamma e^{- 2 \alpha T} \bigg( 1 + \sum_{b \in E}{R_b\theta^2_b} \bigg),
\end{eqnarray*}
where the second equality follows from~\eqref{eq:telescope},
the fourth equality follows from 
$\ind_{b \in p(\rho,v \wedge w)} 
= \ind_{b \in p(\rho,v)} \ind_{b \in p(\rho,w)}$,
and the last equality follows from~\eqref{eq:flow-on-branch}.
\end{proof}
% \begin{remark}
% Note that~\eqref{eq:variance-formula}, which holds
% for a general $(\theta_\ell)_{\ell\in\leaves}$ extended 
% to branches by~\eqref{eq:flow-on-branch}, implies
% that it suffices to consider non-negative $\theta_\ell$'s
% when minimizing $\var[Y_\bftheta]$ under $\|\bftheta\|=1$.
% Indeed assume $(\theta_\ell)_{\ell\in\leaves}$ contains
% negative values and consider the non-negative flow
% $$\bftheta' = \frac{\bftheta_+}{\|\bftheta_+\|},$$
% where $\bftheta_+$ indicates the positive part component-wise.
% Because $\|\bftheta_+\| > 1$, we have $\theta_\ell' < |\theta_\ell|$
% for all leaves $\ell$ and hence
% $\theta_b' < |\theta_b|$
% for all branches $b$.
% \end{remark}

For $0 \leq t \leq T$, let $\pi_t$ be the set of points 
at distance $t$ from the root
(that is, the cutset corresponding to time $t$
away from the root). Noting that
$$
R_b = 2\alpha \int_{\delta(\rho,b)-|b|}^{\delta(\rho,b)}
e^{2\alpha s} \mathrm{d} s,
$$
we get the following convenient formula:
\begin{cor}[Variance formula: Integral form]
\label{cor:integral-form}
For any unit flow $\bftheta$ from $\rho$ to $\mathscr{L}$, we have
\[
\var[Y_\bftheta] = \gamma e^{- 2 \alpha T} \left [ 1 + 
2\alpha \int_0^T e^{2\alpha s} 
\Big(\sum_{x\in \pi_s} \theta_x^2\Big)
\mathrm{d}s \right ].
\]
\end{cor}
\noindent As a first important application of Proposition~\ref{prop:key}
and Corollary~\ref{cor:integral-form}, we show that the
variance of the MLE of $\mu$ can be controlled by the
branching number. The result is characterized
by a transition at $\br=2\alpha$, similarly to Example~\ref{ex:first-transition-continued}.
\begin{prop}[Variance of $\mle$: Link to the branching number]
\label{prop:variance-branching}
Let $\treeseq = (\tree_k)_k$ be a tree sequence
with branching number $\br > 0$.
Then, for all $\Lambda < \br$, there is
$\ical_\Lambda$ such that
\begin{displaymath}
\var_{\tree_k}[\mle^{(k)}]
\leq 
\begin{cases}
\gamma\Big(1 + \frac{2\alpha}{\ical_\Lambda(2\alpha-\Lambda)}\Big)
e^{- \Lambda T_k}, &\text{if $\Lambda < 2\alpha$,}\\
\gamma\;\Big(1 + \frac{2\alpha T_k}{\ical_\Lambda}\Big)
e^{-2\alpha T_k}, &\text{if $\Lambda = 2\alpha$,}\\
\gamma\Big(1 + \frac{2\alpha}{\ical_\Lambda(\Lambda-2\alpha)}\Big)
e^{-2\alpha T_k}, &\text{if $\Lambda > 2\alpha$.}
\end{cases}
\end{displaymath}
\end{prop}
\begin{proof}
For $\Lambda < \br$, let
$
\ical_\Lambda = \inf_{k, \pi \in \Pi^k} 
\sum_{x \in \pi}{ e^{- \Lambda \delta_k(\rho,x)}  > 0 }.
$
By the max-flow min-cut theorem (see e.g.~\cite{lawler1976combinatorial}), there is a flow
$\bfeta^{(k)}$ on $\tree_k$ with 
\begin{equation}\label{eq:max-flow-value}
\|\bfeta^{(k)}\| \geq \ical_\Lambda
\end{equation}
and
\begin{equation}\label{eq:max-flow-capacity}
\eta^{(k)}_x \leq e^{- \Lambda \delta_k(\rho,x)}, 
\end{equation}
for all points $x$ in $\tree_k$. Normalize $\bfeta^{(k)}$ as
$\bftheta^{(k)} = \bfeta^{(k)}/\|\bfeta^{(k)}\|$.
By Proposition~\ref{prop:key}
and Corollary~\ref{cor:integral-form},
for $\Lambda \neq 2\alpha$,
\begin{eqnarray*}
\var_{\tree_k}[\mle^{(k)}]
&\leq& 
\gamma e^{- 2 \alpha T_k} \left [ 1 + 
2\alpha \int_0^{T_k} e^{2\alpha s} 
\Big(\sum_{x\in \pi^k_s} (\theta^{(k)}_x)^2\Big)
\mathrm{d}s \right ]\\
&\leq& 
\gamma e^{- 2 \alpha T_k} \left [ 1 + 
2\alpha \int_0^{T_k} e^{2\alpha s} 
\Big(\sum_{x\in \pi^k_s} \theta^{(k)}_x \frac{e^{- \Lambda \delta_k(\rho,x)}}{\ical_\Lambda}\Big)
\mathrm{d}s \right ]\\
&\leq& 
\gamma e^{- 2 \alpha T_k} \left [ 1 + 
\frac{2\alpha}{\ical_\Lambda}
\int_0^{T_k} e^{( 2\alpha -\Lambda)s} 
\mathrm{d}s \right ]\\
%&=& 
%\gamma e^{- 2 \alpha T_k} \left [ 1 + 
%\frac{2\alpha}{\ical_\Lambda (2\alpha-\Lambda)}
%(e^{(2\alpha -\Lambda)T_k} - 1) 
% \right ]\\
&=& 
\gamma \left [ e^{- 2 \alpha T_k} + 
\frac{2\alpha}{\ical_\Lambda (2\alpha-\Lambda)}\big(e^{- \Lambda T_k} - e^{- 2 \alpha T_k}\big) 
 \right ].
\end{eqnarray*}
where the second inequality follows
from~\eqref{eq:max-flow-value} 
and~\eqref{eq:max-flow-capacity}, and
the third inequality follows from the fact that
$\delta_k(\rho,x) = s$ for $x \in \pi^k_s$
by definition and that
$\sum_{x\in \pi^k_s} \theta^{(k)}_x = 1$.
Similarly if $\Lambda = 2\alpha$
$$
\var_{\tree_k}[\mle^{(k)}]
\leq \gamma \left[ e^{- 2 \alpha T_k} + 
\frac{2\alpha e^{- 2 \alpha T_k}  T_k}{\ical_\Lambda}
\right].
$$
\end{proof}

\paragraph{Removing bottlenecks}
Examining~\eqref{eq:variance-formula},
one sees that a natural bound on $\var[Y_\bftheta]$ is obtained
by ``splitting an edge'' in $\tree$.
\begin{definition}[Edge splitting]
\label{def:splitting}
Let $\tree$ be an ultrametric tree with edge set $E$.
Let $b_0 = (x_0,y_0)$ be a branch in $\tree$
(where $x_0$ is closer to the root)  
and let $b_i = (y_0,y_i)$, $i = 1,\ldots, D$, 
be the outgoing edges at $y_0$. The operation
of {\em splitting branch $b_0$} to obtain a new tree $\tree'$
with edge set $E'$
is defined as follows:
remove $b_0, b_1,\ldots, b_D$ from $\tree$;
add $D$ new edges $b'_i = (x_0,y_i)$ of length
$|b_0| + |b_i|$, $i=1,\ldots,D$ (see Figure \ref{fig:fig2}).
We call {\em merging} the opposite operation
of undoing the above splitting.
\end{definition}
\begin{figure}[h]
\begin{center}
 \scalebox{0.8}{%
\begin{tikzpicture}[domain=0:1,xscale=3.7,yscale=0.7] %,style={font=\footnotesize}]

\draw (0,3.1) node[above] {\large $x_0$};
\draw (0.45,3.1) node[above] {\large $y_0$};
\draw (1.05,5) node[right] {\large $y_1$};
\draw (1.05,4) node[right] {\large $y_2$};
\draw (1.05,1) node[right] {\large $y_D$};
\draw (1,2) node[] {\huge .};
\draw (1,2.3) node[] {\huge .};
\draw (1,2.6) node[] {\huge .};
\draw[line width=1pt] plot coordinates {(0,3) (0.5, 3)};
\draw[line width=1pt] plot coordinates {(0.5,3) (1, 1)};
\draw[line width=1pt] plot coordinates {(0.5,3) (1, 4)};
\draw[line width=1pt] plot coordinates {(0.5,3) (1, 5)};

\draw[->, thick] (1.3,3) -- (1.5,3);

\draw (1.7,3.1) node[above] {\large $x_0$};
\draw (2.75,5) node[right] {\large $y_1$};
\draw (2.75,4) node[right] {\large $y_2$};
\draw (2.75,1) node[right] {\large $y_D$};
\draw (2.7,2) node[] {\huge .};
\draw (2.7,2.3) node[] {\huge .};
\draw (2.7,2.6) node[] {\huge .};

\draw[line width=1pt] plot coordinates {(1.7,3) (2.7, 1)};
\draw[line width=1pt] plot coordinates {(1.7,3) (2.7, 4)};
\draw[line width=1pt] plot coordinates {(1.7,3) (2.7, 5)};

\draw[dashed] plot coordinates {(1.7,3) (2.2, 3)};
\draw[dashed] plot coordinates {(2.2,3) (2.7, 1)};
\draw[dashed] plot coordinates {(2.2,3) (2.7, 4)};
\draw[dashed] plot coordinates {(2.2,3) (2.7, 5)};

\end{tikzpicture}
}
\end{center}
\caption{Edge splitting procedure.}
\label{fig:fig2}
\end{figure}
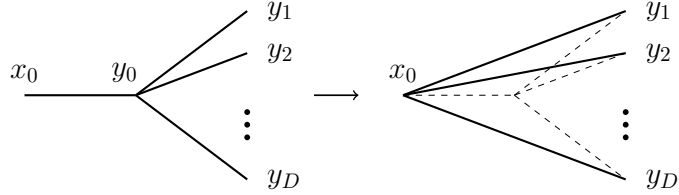
\noindent Note that the number of tips in $\tree$ and $\tree'$
above are the same, and therefore we can use the
same estimator $Y_\bftheta$ on both of them.
\begin{lemma}[Splitting an edge]
\label{lemma:splitting}
Let $\tree$ be an ultrametric tree, let $b_0$ be a branch
in $\tree$, and let $\tree'$ be obtained from $\tree$
by splitting $b_0$. Then for any nonnegative
$\bftheta = (\theta_\ell)_{\ell\in\leaves}$
$$
\var_{\tree'}[Y_\bftheta] \leq \var_{\tree}[Y_\bftheta].
$$
\end{lemma}
\begin{proof}
We use the notation of Definition~\ref{def:splitting}.
Denote by $(\theta_b)_{b\in E}$ and $(\theta'_{b})_{b\in E'}$
the flows associated to $\bftheta$ by~\eqref{eq:flow-on-branch}
on $\tree$ and $\tree'$ respectively.
For any branch $b$, except $b_0,b_1,\ldots,b_D$
and $b'_1,\ldots,b'_D$, we have $\theta_b = \theta'_b$, as the
descendant leaves of $b$ on $\tree$ and $\tree'$
are the same. Think of $b'_i = (x_0,y_i)$, $i=1,\ldots,D$, as being
made of two consecutive edges $b''_i = (x_0,y_i')$ and $b'''_i = (y'_i,y_i)$
with $|b''_i| = |b_0|$ and $|b'''_i| = |b_i|$
(and note, for sanity check, that $R_{b'_i} = R_{b''_i}  + R_{b'''_i}$). 
Then, $\theta_{b_i} = \theta_{b'''_i}$ and 
$R_{b_i} = R_{b'''_i}$, and by~\eqref{eq:variance-formula}
\[ %\begin{eqnarray*}
\frac{\var_{\tree}[Y_\bftheta] - \var_{\tree'}[Y_\bftheta]}{\gamma e^{-2 \alpha T}}
= R_{b_0} \theta_{b_0}^2 - \sum_{i=1}^D R_{b''_i} \theta_{b''_i}^2
= R_{b_0} \Big(\sum_{i=1}^D \theta_{b''_i} \Big)^2 - 
R_{b_0} \sum_{i=1}^D  \theta_{b''_i}^2 \geq 0,
\] %\end{eqnarray*}
where we used that $R_{b_0} = R_{b''_i}$ 
and the nonnegativity of the $\theta_{b''_i}$'s.
\end{proof}
\noindent Comparing $\tree$ to a star we then get:
\begin{prop}[Lower bound on the variance of $\mle$]
\label{prop:variance-lower}
Let $\tree$ be an ultrametric tree with $n$ tips
and height $T$. Then
$$
\var_\tree[\mle] \geq  \gamma
\left(e^{-2\alpha T} + \frac{1 - e^{-2\alpha T}}{n}\right).
$$
\end{prop}
\begin{proof}
Split all edges in $\tree$ by repeatedly applying
Lemma~\ref{lemma:splitting} until a star
tree with $n$ leaves and height $T$ is obtained.
The result then follows from~\eqref{eq:varmlestar}.
\end{proof}

\noindent The following example will be useful 
when proceeding in reverse, to find an upper bound on the variance of $\mle$.
\begin{example}[Spherically symmetric trees]
\label{ex:spherical}
Let $\tree$ be a spherically symmetric, ultrametric tree, that is,
a tree such that all vertices at the same graph distance
from the root have the same number of outgoing
edges, all of the same length.
Let $D_h$, $h=0,\ldots,H-1$, 
be the out-degree of vertices at graph distance
$h$ (where $h=0$ and $h=H$ correspond to the root
and leaves respectively) and let 
$\tau_h$ be the corresponding branch length.
Notice that $\beta_1^2+\cdots+\beta_d^2$,
subject to $\beta_1+\cdots+\beta_d = 1$,
is minimized at $\beta_1=\cdots=\beta_d = 1/d$.
Hence, %using Lemma~\ref{lemma:variational-formulation} 
since $\mle$ is the best unbiased linear estimator % of $\mu$, or: of minimum variance
and arguing inductively from the leaves in~\eqref{eq:variance-formula}, we see that
$\mle=\overline{Y}$ in this case. The mean squared error
is, by~\eqref{eq:variance-formula}, 
\begin{eqnarray}
\var[\mle] 
&=& \gamma e^{-2\alpha T} \left[
1 + \sum_{h=0}^{H-1} \left(\prod_{h'=0}^{h} D_{h'}\right) (1 - e^{-2\alpha \tau_h})
e^{2\alpha \sum_{h'=0}^{h} \tau_{h'}} 
\prod_{h'=0}^{h}\frac{1}{D_{h'}^2} 
\right]\nonumber\\
&=& \gamma e^{-2\alpha T} \left[
1 + \sum_{h=0}^{H-1} (1 - e^{-2\alpha \tau_h})
 \prod_{h'=0}^{h}\frac{e^{2\alpha \tau_{h'}}}{D_{h'}} 
\right].\label{eq:spherical}
\end{eqnarray}
\end{example}

\begin{prop}[Upper bound on the variance of $\mle$]
\label{prop:variance-upper}
Let $\tree$ be an ultrametric tree with height $T$. Recall that
$\pi_t$ be the set of points at distance $t$ from the root. Then
$$
\var_\tree[\mle]
\leq
\inf_{0 \leq t \leq T}
\gamma 
\left(
e^{-2\alpha (T-t)} 
+ \frac{1 - e^{-2\alpha(T-t)}}{|\pi_t|}
\right).
$$
\end{prop}
\begin{proof}
Let $0 \leq t \leq T$.
For all points $x$ in $\pi_t$, choose one descendant
leaf $\ell_x$ of $x$ and define $\bftheta$ as
\begin{displaymath}
\theta_\ell
=
\begin{cases}
1/|\pi_t| & \text{if $\ell = \ell_x$ for some $x$},\\
0         & \text{otherwise}.
\end{cases}
\end{displaymath}
Divide all branches crossing $\pi_t$
into two branches meeting at $\pi_t$.
Then merge all branches above $\pi_t$ 
(that is, closer to the root) by repeatedly
applying Lemma~\ref{lemma:splitting}.
By~\eqref{eq:variance-formula}, removing
all branches $b$ with $\theta_b = 0$ does not
affect the variance, and from Example~\ref{ex:spherical}
with $H = 2$, $D_0 = 1$, $D_1 = |\pi_t|$, $\tau_0 = t$,
and $\tau_1 = T-t$, we get
\begin{eqnarray*}
\var_\tree[\mle]
&\leq& 
\gamma e^{-2\alpha T}
\left[
1+ (1 - e^{-2\alpha t})e^{2\alpha t} 
+ (1 - e^{-2\alpha(T-t)}) e^{2\alpha t} \frac{e^{2\alpha (T-t)}}{|\pi_t|}
\right]\\
&\leq& 
\gamma 
\left[
e^{-2\alpha (T-t)} 
+ \frac{1 - e^{-2\alpha(T-t)}}{|\pi_t|}
\right].
\end{eqnarray*}
\end{proof}

\paragraph{The two estimators $\mle$ vs.~$\overline{Y}$}
As an application of the previous proposition,
we provide an example where $\mle$ performs
significantly better than $\overline{Y}$.
Roughly, the example shows that $\overline{Y}$
can perform poorly on asymmetric trees.
\begin{example}
\label{ex:caterpillar-bad}
Consider a caterpillar sequence $(\mathbb{T}_k)_{k}$,
as defined in Example~\ref{ex:caterpillar}, with
$t_{2m+1} = m$ and $t_{2m} = 1$ 
for all $m$, as shown in Figure~\ref{fig:fig1}.
Note that the tree height is $T_{2m+1} = T_{2m+2} = m$
and the cut sets $\pi_t^{k}$ of $\tree_k$ at time $t$ 
satisfy $|\pi_{m-1}^{2m+1}| = |\pi_{m-1}^{2m+2}| = m$. 
Therefore, by Proposition~\ref{prop:variance-upper}, 
$$
\max\left\{\var_{\tree_{2m+1}}[\mle],
\var_{\tree_{2m+2}}[\mle]\right\}
\leq \gamma \left[
e^{-2\alpha (m-1)} + \frac{1}{m} 
\right] \to 0,
$$
as $m \to +\infty$, and hence
$\mle$ is consistent. 
On the other hand, note that 
$\cov[Y_i,Y_j] \geq 0$ 
for all pairs of leaves $i, j$ in $\tree_k$. 
Therefore,
\begin{eqnarray*}
\var_{\tree_{2m}}\left[\overline Y\right] 
&=& \frac{1}{4m^2} \var\left[
\sum_{\ell=1}^{2m}{Y_\ell} \right] %\\ 
\geq \frac{1}{4m^2} \var \left[
\sum_{i=1}^m{Y_{2i}} \right] %\\
= \frac{1}{4m^2} \sum_{i,j=1}^m 
\cov[Y_{2i},Y_{2j}]\\ 
&\geq& \frac{1}{4m^2} m^2 \gamma e^{-2 \alpha}
= \frac{\gamma e^{- 2 \alpha}}{4}.
\end{eqnarray*}
So, $\overline Y$ is not consistent.
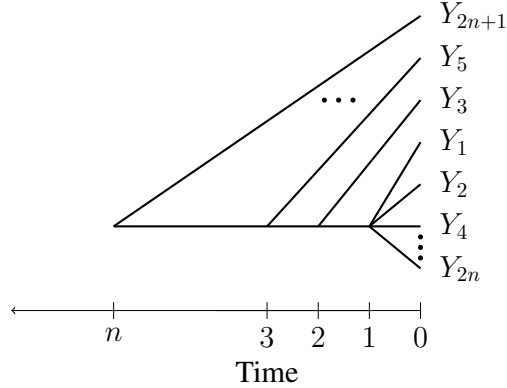
\begin{figure}[h]
\begin{center}
 \scalebox{0.8}{%
\begin{tikzpicture}[domain=0:1,xscale=1.7,yscale=0.7] %,style={font=\footnotesize}]

\draw[very thin,color=black] (1.5,-0.2) grid[xstep=0.5] (3.5,0.2); % grid on X axis
\draw[very thin,color=black] (0.5,-0.2) grid[xstep=0.5] (0.5,0.2); % grid on X axis
\draw[->] (3,0) -- (-0.5,0); % X axis

\draw (3.5,-0.2) node[below] {\large $0$};
\draw (3,-0.2) node[below] {\large $1$};
\draw (2.5,-0.2) node[below] {\large $2$};
\draw (2,-0.2) node[below] {\large $3$};
\draw (0.5,-0.2) node[below] {\large $n$};
\draw (2,-1) node[below] {\large Time};

\draw (3.6,1) node[right] {\large $Y_{2n}$};
\draw (3.5,1.25) node[] {\huge $.$};
\draw (3.5,1.5) node[] {\huge $.$};
\draw (3.5,1.75) node[] {\huge $.$};
\draw (3.6,2) node[right] {\large $Y_4$};
\draw (3.6,3) node[right] {\large $Y_2$};
\draw (3.6,4) node[right] {\large $Y_1$};
\draw (3.6,5) node[right] {\large $Y_3$};
\draw (3.6,6) node[right] {\large $Y_5$};
\draw (3.6,7) node[right] {\large $Y_{2n+1}$};
\draw (2.7,5) node[] {\huge $...$};

\draw[line width=1pt] plot coordinates {(0.5,2) (3.5, 2)};
\draw[line width=1pt] plot coordinates {(3,2) (3.5, 1)};
\draw[line width=1pt] plot coordinates {(3,2) (3.5, 3)};
\draw[line width=1pt] plot coordinates {(3,2) (3.5, 4)};
\draw[line width=1pt] plot coordinates {(2.5,2) (3.5, 5)};
\draw[line width=1pt] plot coordinates {(2,2) (3.5, 6)};
\draw[line width=1pt] plot coordinates {(0.5,2) (3.5, 7)};

\end{tikzpicture}
}
\end{center}
\caption{Example where the MLE $\hat \mu$ is consistent while $\overline Y$ is not.}
\label{fig:fig1}
\end{figure}
\end{example}

\subsection{Proof of Theorem~\ref{thm:criterion-consistency} and Sufficiency of Conditions}
\label{sec:consistency-mle}

\begin{proof}[Proof of Theorem~\ref{thm:criterion-consistency} (Consistency of $\mle$)]
First assume~\eqref{eq:consistency-condition}.
From Proposition~\ref{prop:variance-upper}, for all $s$,
$$
\limsup_k \var_{\tree_k}[\mle^{(k)}]
\leq
\limsup_k \gamma 
\left[
e^{-2\alpha s} 
+ \frac{1 - e^{-2\alpha s}}{\left|\tilde{\pi}^k_s\right|}
\right]
\leq \gamma e^{-2\alpha s}.
$$
Taking $s$ to $+\infty$ gives consistency.
On the other hand, assume by contradiction 
that $(\mle^{(k)})_k$ is consistent
but that
$ \liminf_k \left|\tilde{\pi}^k_{s}\right| = L < +\infty $
for some $s \in (0, +\infty)$. Let $(k_j)_j$ be the corresponding
subsequence. 
Divide all branches in $\tree_{k_j}$ 
crossing $\tilde{\pi}^{k_j}_s$
into two branches meeting at $\tilde{\pi}^{k_j}_s$.
Split edges 
in $\tree_{k_j}$ above $\tilde{\pi}^{k_j}_s$ 
(closer to the root) repeatedly 
until the tree above $\tilde{\pi}^{k_j}_s$ forms a star.
Let $\tree'$ be the resulting tree,
let $b'_1,\ldots,b'_{D}$ be the branches emanating
from the root, where $D \leq L$ by assumption,
and let $\tilde{\pi}'$ be the cutset at time $s$ from the leaves. For the unit flow $\bftheta'$ corresponding to the MLE on $\tree'$,
by Lemma~\ref{lemma:splitting} 
and counting only those edges above $\tilde{\pi}'$
in $\tree'$
in~\eqref{eq:variance-formula},
we have
\begin{eqnarray*}
\var_{\tree_{k_j}}[\mle^{(k_j)}]
&\geq& \gamma e^{-2\alpha T_{k_j}}
\left[
1 + \left(1 - e^{-2\alpha (T_{k_j} - s)}\right)e^{2\alpha (T_{k_j} - s)}
\sum_{i=1}^D (\theta'_{b'_i})^2
\right]\\
%&\geq& \gamma e^{-2\alpha T_{k_j}}
%\left[
%1 + \left(1 - e^{-2\alpha (T_{k_j} - s)}\right) \frac{e^{2\alpha (T_{k_j} - s)}}{D}
%\right]\\
&\geq& \gamma 
\left[
e^{-2\alpha T_{k_j}} + \left(1 - e^{-2\alpha (T_{k_j} - s)}\right) \frac{e^{- 2\alpha s}}{L}
\right],
\end{eqnarray*}
where we used the fact that $\beta_1^2+\cdots+\beta_D^2$,
subject to $\beta_1+\cdots+\beta_D = 1$,
is minimized at $\beta_1=\cdots=\beta_D = 1/D$.
Since $T_{k_j} \to +\infty$ under Assumption~\ref{assumption:unbounded},
$$
\limsup_k \var_{\tree_{k}}[\mle^{(k)}]
\geq \gamma \frac{e^{- 2\alpha s}}{L} > 0,
$$
and we get a contradiction.
\end{proof}
%More generally, Assumption~\ref{assumption:unbounded}
%does not suffice for the consistency of $\mle$, as the following example shows.
%\begin{example}[Consistency: Counter-example]
%\label{ex:spherical-bad}
%Using the notation of Example~\ref{ex:spherical},
%let $(\tree_{k})_k$ be a sequence of spherically
%symmetric trees with $H=2$ levels of nodes, 
%degrees $D^{(k)}_0 = 2$ at the root and
%$D^{(k)}_1 = m$ at internal nodes, and branch lengths
%$\tau^{(k)}_0 = k-1$, and $\tau^{(k)}_1 = 1$. 
%Although Assumption~\ref{assumption:unbounded} holds
%($n_{k} = 2k$ and $T_{k} = k$), we have by~\eqref{eq:spherical}
%\begin{eqnarray*}
%\var_{\tree_{k}}[\mle^{(k)}]
%&=& \gamma e^{-2\alpha k}\left[
%1 + (1-e^{-2\alpha (k-1)}) \frac{e^{2\alpha (k-1)}}{2}
%+ (1 - e^{-2\alpha}) \frac{e^{2\alpha k}}{2k}
%\right]\\
%&\geq&\gamma e^{-2\alpha}/2 > 0, 
%\end{eqnarray*}
%and the MLE is not consistent.
%Taking $s = 1$ in~\eqref{eq:consistency-condition}
%explains why. 
%In this example, it is also easy to see that $\ugr=0$. 
%However, in general, the branching number
%does provide a simple, sufficient condition.
%\end{example}
\noindent We note that, by Proposition~\ref{prop:variance-branching}, the branching number
provides a simple, sufficient condition for consistency.
\begin{prop}[Consistency: Branching number condition]
\label{thm:consistency-branching}
Let $\treeseq = (\tree_k)_k$ be a tree sequence
satisfying Assumption~\ref{assumption:unbounded}
with branching number $\br$. Then
$\br > 0$ suffices for the consistency of the
MLE of $\mu$.
%Note that this condition is independent of $\alpha$.
\end{prop}
%\begin{proof}
%By Proposition~\ref{prop:variance-branching},
%taking $0 < \Lambda < \min\{2\alpha,\br\}$,
%$$
%\var_{\tree_k}[\mle^{(k)}]
%\leq
%\gamma\left[1 + \frac{2\alpha}{\ical_\Lambda(2\alpha-\Lambda)}\right]
%e^{- \Lambda T_k} \to 0,
%$$
%as $k \to +\infty$.
%\end{proof}
%\noindent The condition in Proposition~\ref{thm:consistency-branching}
%is not necessary, as the following example shows.
%\begin{example}[Caterpillar sequence: Consistency]
%\label{ex:caterpillar-consistency}
%Using the notation of Example~\ref{ex:caterpillar},
%let $\treeseq = (\tree_k)_k$ be a caterpillar sequence
%with $t_k = \omega(\log k)$. 
%The MLE of $\mu$
%is consistent on $\treeseq$ by Corollary~\ref{cor:consistency-nested}.
%However let $\Lambda > 0$. For $k = 1,2,\ldots$
%\begin{eqnarray*}
%\inf_{\pi \in \Pi^k}
%\sum_{x \in \pi} e^{- \Lambda \delta_k(\rho,x)}
%&\leq& k e^{- \Lambda \omega(\log k)} \to 0,
%\end{eqnarray*}
%as $k\to +\infty$, where we used that $T_k \geq t_k$. 
%Therefore $\br = \ugr = 0$.
%\end{example}

\subsection{Phase transition on the rate of convergence of the MLE}
\label{sec:convergence-rate-mle}

Theorems~\ref{thm:rate-supercritical} and~\ref{thm:subcritical}
show a phase transition for the
$\sqrt{n_k}$-consistency of $\mle$, which we prove now. 

\begin{proof}[Proof of Theorem~\ref{thm:rate-supercritical} (Supercritical regime)]
Assume $\ugr  > 2\alpha$.
As remarked after Definition~\ref{def:growth},
for all $\epsilon > 0$, eventually
\begin{equation}\label{eq:nk-bound}
\exp\left( 
(\lgr - \epsilon)
T_k\right)
\leq n_k \leq \exp\left( 
(\ugr + \epsilon) 
T_k\right),
\end{equation}
that is,
$
n_k^{-2\alpha/(\lgr - \epsilon)}
\leq e^{-2\alpha T_k} 
\leq 
n_k^{-2\alpha/(\ugr + \epsilon)}.
$
Moreover for all $\epsilon> 0$
there are subsequences $(k_j)_j$ and $(k'_j)_j$
such that
\begin{equation}\label{eq:subseq-supercritical}
n_{k_j} \geq \exp\left( 
(\ugr - \epsilon) 
T_{k_j}\right)
\mbox{ and }
n_{k'_j} \leq 
\exp\left( 
(\lgr + \epsilon) 
T_{k'_j}\right).
\end{equation}
By Proposition~\ref{prop:variance-lower},
\begin{equation}\label{eq:super-lower}
\var_{\tree_k}[\mle^{(k)}]
\geq \gamma\left[
e^{-2\alpha T_k}
+ \frac{1 - e^{-2\alpha T_k}}{n_k}
\right] \geq \gamma e^{-2\alpha T_k}.
\end{equation}
Then~\eqref{eq:non-sqrt-supercritical} follows
from~\eqref{eq:subseq-supercritical}
and~\eqref{eq:super-lower}.
Hence $n_k \var_{\tree_k}[\mle^{(k)}] \to +\infty$ along a subsequence and
$(\mle^{(k)})_k$ is not $\sqrt{n_k}$-consistent
(using that $\mle$ is unbiased and normally distributed). 
% by Lemma~\ref{lemma:not-consistent}.

Assume $\br > 2\alpha$.
Let $2\alpha < \Lambda < \br$. 
By Proposition~\ref{prop:variance-branching}
\begin{equation}\label{eq:super-upper}
\var_{\tree_k}[\mle^{(k)}]
\leq \gamma\left[
1 + \frac{2\alpha}{\ical_\Lambda (\Lambda - 2\alpha)}
\right]
e^{-2\alpha T_k}.
\end{equation}
Note that $\ugr \geq \lgr \geq \br > 2\alpha$ and
hence, by~\eqref{eq:super-lower} and~\eqref{eq:super-upper},
$
\var_{\tree_k}[\mle^{(k)}] = \Theta\left(e^{-2\alpha T_k}\right).
$
Combining this with~\eqref{eq:nk-bound}
gives the result in terms of $n_k$.

Assume instead that $\br < 2\alpha$.
Let $\Lambda < \br$. 
By Proposition~\ref{prop:variance-branching}
\begin{equation*} %\label{eq:super-upper-second}
\var_{\tree_k}[\mle^{(k)}]
\leq \gamma\left[
1 + \frac{2\alpha}{\ical_\Lambda (2\alpha-\Lambda)}
\right]
e^{- \Lambda T_k}.
\end{equation*}
The rest of the argument is similar to the previous case.
\end{proof}

\begin{proof}[Proof of Proposition~\ref{prop:example2level}]
Note that Example~\ref{ex:twolevel} considers a
spherically symmetric tree.
%as defined in Example~\ref{ex:spherical},
%with $H=2$ levels. 
By~\eqref{eq:spherical},
\[
\var_{\tree_k}[\mle^{(k)}] 
= \gamma e^{-2\alpha (\tau^{(k)}_0 + \tau^{(k)}_1)} \left[
1 + \sum_{h=0,1} (1 - e^{-2\alpha \tau^{(k)}_h})
 \prod_{h'=0}^{h}\frac{e^{2\alpha \tau^{(k)}_{h'}}}{D^{(k)}_{h'}} 
\right]\]
which then gives \eqref{eq:variance-2-level}.
Note that
$$
\frac{\log n_k}{ T_k}
= \frac{\Lambda_0\tau^{(k)}_0 + \Lambda_1 \tau^{(k)}_1}
{\tau^{(k)}_0 + \tau^{(k)}_1}
= \Lambda_0 \sigma + \Lambda_1 (1-\sigma)
= \gr.
$$
To compute the branching number, it suffices
to consider cutsets with $m_0$ level-1 vertices
and the $D^{(k)}_1(D^{(k)}_0 - m_0)$ tips below the rest of the
level-1 vertices. Then
$$
\jcal_k 
\equiv
\inf_{\pi \in \Pi^k}\ 
\sum_{x \in \pi} e^{- \Lambda \delta_k(\rho,x)}
=
\begin{cases}
D^{(k)}_0 e^{-\Lambda \tau^{(k)}_0},
& \text{if $D^{(k)}_1 > e^{\Lambda \tau^{(k)}_1}$}\\
n_k e^{- \Lambda T_k}, & \text{otherwise.}
\end{cases}
$$
Hence if $\Lambda \geq \Lambda_1 > \gr$
we are in the second case and 
$ n_k e^{- \Lambda T_k} = e^{- (\Lambda-\gr) T_k} \to 0,$
as $k \to +\infty$.
If $\Lambda < \Lambda_1$
we are in the first case and 
$ D^{(k)}_0 e^{-\Lambda \tau^{(k)}_0} = e^{- (\Lambda-\Lambda_0) \tau^{(k)}_0},$
so that $\br = \Lambda_0$.
\end{proof}

\begin{proof}[Proof of Theorem~\ref{thm:subcritical} (Subcritical regime)]
One direction follows immediately from
Pro\-position~\ref{prop:variance-lower} which implies
$$
\var_{\tree_k}[\mle^{(k)}] \geq \gamma \frac{1 - e^{-2\alpha T_k}}{n_k}
= \Omega(n_k^{-1}).
$$
We prove the other direction separately in each case.
Assume first that $0 < \br = \ugr < 2\alpha$.
For $\epsilon > 0$ (small), choose $\Lambda$ such that
$$
\ugr - \epsilon < \Lambda < \ugr = \br < 2\alpha.
$$
By Proposition~\ref{prop:variance-branching},
eventually
\[ %\begin{eqnarray*}
\var_{\tree_k}[\mle^{(k)}]
\leq
\gamma\left[1 + \frac{2\alpha}{\ical_\Lambda(2\alpha-\Lambda)}\right]
e^{- \Lambda T_k}
\leq 
\gamma\left[1 + \frac{2\alpha}{\ical_\Lambda(2\alpha-\Lambda)}\right]
n_k^{-(\ugr -\epsilon)/(\ugr + \epsilon)}.
\] %\end{eqnarray*}

Assume instead that $\unigr < 2\alpha$.
We show that $\overline{Y}$ 
(and hence the MLE by Proposition~\ref{prop:key}) achieves
$\sqrt{n_k}$-consistency in this case.
Let $\bftheta$ be the corresponding flow on $\tree_k$.
By Corollary~\ref{cor:integral-form},
letting $\unigr < \Lambda < 2\alpha$,
for $k$ large enough
\begin{eqnarray*}
\var_{\tree_k}[\overline{Y}] 
&=& 
\gamma e^{- 2 \alpha T_k} \left [ 1 + 
2\alpha \int_0^{T_k} e^{2\alpha s} 
\left(\sum_{x\in \pi^k_s} \left(\frac{n_k(x)}{n_k}\right)^2\right)
\mathrm{d}s \right ]\\
&\leq&
\gamma e^{- 2 \alpha T_k} \left [ 1 + 
2\alpha \int_0^{T_k} e^{2\alpha s} 
\left(\sum_{x\in \pi^k_s} \left(\frac{n_k(x)}{n_k}\right)
\frac{e^{\Lambda [(T_k-s)+M]}}{n_k}\right)
\mathrm{d}s \right ]\\
&\leq&
%\gamma e^{- 2 \alpha T_k} \left [ 1 + e^{\Lambda M}
%\frac{2\alpha}{n_k (2\alpha-\Lambda)} \,
%e^{\Lambda T_k} (e^{(2\alpha -\Lambda) T_k} - 1)
%\right ]\\
%&=&
\gamma  \left [ e^{- 2 \alpha T_k} + 
e^{\Lambda M}
\frac{2\alpha}{n_k (2\alpha-\Lambda)} \,
(1 - e^{-(2\alpha -\Lambda) T_k})
\right ].
\end{eqnarray*}
The result follows from the fact that
$e^{ \Lambda [T_k+M]} \geq n_k$.
\end{proof}

\subsection{Proofs for special cases}
\label{section:proofs-special}

\begin{proof}[Proof of Corollary~\ref{cor:mu-yule}]
	By Theorems~\ref{thm:criterion-consistency},~\ref{thm:rate-supercritical}, and~\ref{thm:subcritical}, it suffices
	to prove that $\br = \ugr = \lambda$
	with probability $1$.
	A Galton-Watson (GW) branching process is
	a discrete-time non-negative integer-valued
	population process defined as follows: at each time step, each
	individual in the population has an independent number
	of offsprings, according to a distribution $F$, that form
	the population at the next time. 
	In~\cite[Chapter 3]{peres1999climb}, it is shown that
	a GW tree where $F$ has mean $m$ has
	branching number and
	upper growth equal to
	$\log m$.
	
	To compute the branching number of an infinite Yule tree $\tree_0$,
	we use a comparison to a GW tree.
	Fix $\epsilon > 0$.
	Let $F$ be the distribution of the number of 
	lineages in $\tree_0$ at time $\epsilon$.
	By standard branching process 
	results~\cite[Equation (4) on page 108]{athreya2004branching},
	$m = e^{\lambda \epsilon}$. By the memoryless property
	of the exponential, the number of lineages $|\pi_{N\epsilon}|$ in the
	Yule tree at time $N\epsilon$ is identically distributed to
	the population size $Z_{N}$ 
	of a GW tree with offspring
	distribution $F$ at time $N$. 
	Then
	$$
	\frac{\log |\pi_s|}{s} 
	\leq \frac{\log Z_{\lceil s/\epsilon \rceil}}{s}
	= \frac{\lceil s/\epsilon \rceil}{ s} \cdot 
	\frac{\log Z_{\lceil s/\epsilon \rceil}}{\lceil s/\epsilon \rceil},
	$$
	which implies that 
	$
	\ugr \leq \frac{1}{\epsilon} \cdot
	\log e^{\lambda \epsilon} = \lambda.
	$
	
	Similarly, let $\pi$ be a cutset in $\tree_0$ and
	let $\pi_\epsilon$ be the cutset obtained 
	by rounding up the points in $\pi$ to the next
	$\epsilon$-multiple closer to the root (removing duplicates). 
	Let $\delta_{\mathrm{GW}}(v)$ be the distance
	from the root to vertex $v$ in the GW tree.
	Then
	$$
	\sum_{x \in \pi} 
	e^{- \Lambda \delta_0(\rho,x)}
	\geq \sum_{y \in \pi_\epsilon} 
	e^{- \Lambda (\delta_{\mathrm{GW}}(y)+1)\epsilon}
	= e^{- \Lambda \epsilon} \sum_{y \in \pi_\epsilon} 
	e^{- (\epsilon \Lambda) \delta_{\mathrm{GW}}(y)}
	> 0
	$$
	whenever $\epsilon \Lambda < \log e^{\lambda \epsilon}$,
	so that
	$
	\br \geq \lambda.
	$
\end{proof}

\begin{proof}[Proof of Corollary~\ref{cor:yule_alpha_gamma}]
	Let $\tau_i = t_{i} - t_{i-1}$ be the amount of time 
	during which $\tree_0$ has $i$ lineages (with $t_0 = 0$). 
	Then $(\tau_i)_i$ are independent exponential 
	random variables with parameters $(1/(i\lambda))_i$.
	Let $T^j_i =\sum_{r=i+1}^j \tau_r$. 
	Note that 
	$$ %\begin{equation}\label{eq:yule-alpha-expec}
	\E{T^j_i} 
	= \sum_{r=i+1}^j {\E{\tau_r}}
	\equiv \lambda^{-1} \sum_{r=i+1}^j \frac{1}{r}
	\in \left(\lambda^{-1}\log\Big(\frac{j}{i+1}\Big), \lambda^{-1}\log\Big(\frac{j}{i}\Big)\right).
	$$ 
	% with
	% $$\log\left(\frac{j}{i+1}\right) = \int_{i+1}^j \frac{1}{x} \mathrm{d} x 
	% \leq\sum_{r=i+1}^j \frac{1}{r} \leq \int_{i}^j \frac{1}{x} \mathrm{d} x
	% \leq \log\left(\frac{j}{i}\right). $$
	Similarly,
	\begin{equation}
	\label{eq:yule-alpha-var}
	\var[T^j_i] 
	= \sum_{r=i+1}^j \var[\tau_r] 
	= \lambda^{-2}
	\sum_{r=i+1}^j \frac{1}{r^2} \leq \frac{1}{\lambda^2 i}.
	\end{equation}
	% with
	% $$ \sum_{r=i+1}^j \frac{1}{r^2} \leq \int_{i}^{+\infty} \frac{1}{x^2} \mathrm{d} x 
	% = \frac{1}{i}. $$
	By Chebyshev's inequality, for all $0<\sigma<1$,
	$$
	\P{
		|T_{\lfloor \sigma k \rfloor}^k
		- \E{T_{\lfloor \sigma k \rfloor}^k}|
		\geq \epsilon
	} = O(k^{-1}),
	$$
	where we used~\eqref{eq:yule-alpha-var}.
	Let $0 < \sigma_2' < \sigma_2
	< \sigma_1' < \sigma_1 < 1$.
	From the previous equation, we get
	for $\iota=1,2$
	$$
	\P{
		T_{\lfloor \sigma_\iota k \rfloor}^k
		\leq \lambda^{-1}\log\left(\frac{k}{\lfloor \sigma_\iota k\rfloor + 1}\right) - \epsilon
	} = O(k^{-1}),
	$$
	and similarly for the other direction.
%	$$
%	\P{
%		T_{\lfloor \sigma_\iota' k \rfloor}^k
%		\geq \lambda^{-1}\log\left(\frac{k}{\lfloor \sigma_\iota' k\rfloor}\right) + \epsilon
%	} = O(k^{-1}).
%	$$
	Take
	$$
	a_\iota
	= \lambda^{-1}
	\log\left(\frac{1}{\sigma_\iota}\right),
	\quad
	a_\iota'
	= \lambda^{-1}
	\log\left(\frac{1}{\sigma_\iota'}\right)
	\quad\mbox{and }
	\epsilon 
	< a_1 % \left[\lambda^{-1}\log\left(\frac{1}{\sigma_1}\right)\right]
	\land
	\frac{1}{2}
	\left[ a_2  %\lambda^{-1}\log\left(\frac{1}{\sigma_2}\right)
	-      a_1' %\lambda^{-1}\log\left(\frac{1}{\sigma_1'}\right)
	\right].
	$$
	Then Assumption~\ref{assump:alpha}
	is satisfied asymptotically with 
	$
	c_\iota=a_\iota-\epsilon$,
	$c_\iota' = a_\iota'+\epsilon$
	and
	$\beta 
	= [\sigma_1 - \sigma_1']
	\land
	[\sigma_2 - \sigma_2'],
	$
	because then
	$$
	\P{c_1 < T_{\lfloor \sigma_1 k \rfloor}^k < T_{\lfloor \sigma_1' k \rfloor}^k< 
		c'_1 <c_2 <
		T_{\lfloor \sigma_2 k \rfloor}^k < T_{\lfloor \sigma_2' k \rfloor}^k<
		c'_2
	}
	\geq 1-O(k^{-1}).
	$$
\end{proof}

\subsection{Sensitivity to estimate of $\alpha$}
\label{sec:sensitivity-to-alpha}

So far in this section, we considered the MLE of $\mu$
{\em given $\alpha$}. Here we look at the sensitivity
of the MLE to estimation errors on $\alpha$. 
Theorem~\ref{thm:main-alpha-rate} shows that there exists
a $\sqrt{n_k}$-consistent estimator of $\alpha$ 
under Assumption~\ref{assump:alpha}, which is unrelated to the growth or 
height of the species tree. 
Moreover the estimator of $\alpha$ we derive
does not require the knowledge of $\mu$.

Hence suppose that we have a $\sqrt{n_k}$-consistent
estimator $\hat\alpha_k$ of $\alpha$.
Let $\evar_{\tree_k}$ denote the variance
under the parameter $\alpha = \hat\alpha_k$ 
(with $\mu$ and $\gamma$ unchanged)
and let $\hat\bftheta_k$ be the corresponding
weights of the MLE of $\mu$, 
that is, the choice of weights 
assuming that $\alpha = \hat\alpha_k$ and minimizing
$\evar_{\tree_k}[Y_{\bftheta}]$. 

For all $k$ and under the true $\alpha$, $Y_{\hat\bftheta_k}$
is an unbiased estimator of $\mu$. 
Moreover, because $\hat\alpha_k = \alpha + o(1)$ 
and so on, the bounds in Theorems~\ref{thm:rate-supercritical}
and~\ref{thm:subcritical} apply to 
$\evar_{\tree_k}[Y_{\hat\bftheta_k}]$ as well
(for $k$ large enough).
The quantity of interest is $\var_{\tree_k}[Y_{\hat\bftheta_k}]$.
By \eqref{eq:variance-formula},
\begin{eqnarray*}
&&\var_{\tree_k}[Y_{\hat\bftheta_k}] \\
&&\quad =   
\gamma e^{- 2 \alpha T_k} 
+ \gamma \sum_{b \in E_k}(1 - e^{-2 \alpha |b|}) 
e^{2 \alpha (\delta_k(\rho,b)-T_k)}
(\hat\theta_k)^2_b \\
&&\quad = (1+O(T_k n_k^{-1/2}))  
\left [\gamma e^{- 2 \hat\alpha_k T_k} 
+ \gamma \sum_{b \in E_k}(1 - e^{-2 \hat\alpha_k |b|}) 
e^{2 \hat\alpha_k (\delta_k(\rho,b)-T_k)}
(\hat\theta_k)^2_b \right ]\\
&&\quad = (1+O(T_k n_k^{-1/2})) \evar_{\tree_k}[Y_{\hat\bftheta_k}],
\end{eqnarray*}
provided $T_k n_k^{-1/2} = o(1)$.
Hence, for instance if $\lgr > 0$, $T_k = O(\log n_k)$
and we get that  $\var_{\tree_k}[Y_{\hat\bftheta_k}]$
satisfies the bounds in Theorems~\ref{thm:rate-supercritical}
and~\ref{thm:subcritical}.

\section{Convergence rate of a new estimator for $\alpha$ and $\gamma$} 
\label{sec:alpha_gamma}

In this section, we provide a novel estimator for
$(\alpha, \gamma)$. Under natural assumptions
on the species tree, we show that this estimator
is $\sqrt{n_k}$-consistent. Moreover this estimator
does not require the knowledge of $\mu$.
Interestingly,
in contrast to what we showed for $\mu$, 
the conditions for $\sqrt{n_k}$-consistency in this
case 
do not involve the growth, or even
the height, of the species tree. 
This is in line with the results in \cite{HoAne13}, 
who found that $\mu$ requires an unbounded tree height
to be microergodic, 
whereas $\alpha$ and $\gamma$ do not.

Note, however, that the MLE
of $\alpha$ and $\gamma$ 
are not simple
linear estimators, which makes them harder to study
here.
In particular, unlike in the case of $\mu$,
we do not provide lower 
bounds on their rate of convergence.

\subsection{Contrast-based estimator}
\label{sec:contrasts}

We first describe the estimator. 
The proof of its convergence rate is 
in Section~\ref{sec:alpha-rate}.

\paragraph{Contrasts}
Our estimator 
relies on an appropriately chosen set of contrasts, that is,
differences between pairs of leaf states (see e.g.~\cite{felsenstein2004inferring}). 
More specifically, we choose contrasts associated 
with internal nodes, as follows. 
Let $\tree$ be an ultrametric species tree with leaves
$\leaves$ and internal vertices $\internal$.
For two leaves $\ell$ and $\ell'$, we let $\ell\land \ell'$
be their most recent common ancestor. 
Assume
that all internal vertices of $\tree$ have out-degree
at least $2$.
Let $i \in \internal$ be an internal vertex of
$\tree$, and 
let $\ell^i_1\neq\ell^i_2$ be two leaves such that
$\ell^i_1 \land \ell^i_2 = i$. 
Let $P_i$ be the path connecting $\ell^i_1$ and $\ell^i_2$. 
We define the corresponding contrast 
$\contrast_i = Y_{\ell^i_1} - Y_{\ell^i_2}$. 
Let $T(i)$ be the height of $i$ from the leaves.
We say that $T(i)$ is the height of $\contrast_i$. 
\begin{lemma}[Contrasts: Distribution~\cite{HoAne13}] 
Let $i_1,\ldots,i_m$ be a collection of internal nodes
of $\tree$. Let $\contrast_{i_1}, \ldots, \contrast_{i_m}$
be an arbitrary set of associated contrasts. 
Assume that the corresponding paths 
$P_{i_1},\ldots,P_{i_m}$ are pairwise non-intersecting, that is,
none of the pairs of paths share a vertex. 
Then
$\contrast_{i_1}, \ldots, \contrast_{i_m}$ 
are mutually independent, multivariate normal with
$\contrast_i \sim {\cal N}(0,2\gamma(1-e^{-2\alpha T(i)}))$.\label{lem_addmic06}
\end{lemma}
\begin{proof}
Indeed, expanding the covariance, we get for $j\neq j'$
\begin{eqnarray*}
\gamma^{-1}\cov[\contrast_j, \contrast_{j'}]
&=& e^{-\alpha d_{\ell^{j}_1\ell^{j'}_1}}
- e^{-\alpha d_{\ell^{j}_1\ell^{j'}_2}}
- e^{-\alpha d_{\ell^{j}_2\ell^{j'}_1}}
+ e^{-\alpha d_{\ell^{j}_2\ell^{j'}_2}} = 0,
\end{eqnarray*}
since, by assumption, 
$\ell^{j}_\iota \land \ell^{j'}_{\iota'}$ is the same vertex
for all $\iota, \iota' = 1,2$.
\end{proof}
\noindent The following lemma will be useful 
in identifying an appropriate collection of contrasts.
\begin{lemma}[Contrasts: A large collection~\cite{HoAne13}]
Let $\mathbb{T}$ be an ultrametric tree 
and let $\internal_{(a,b)}$ 
be the set of internal nodes of $\mathbb{T}$ whose 
height from the leaves lies in $(a,b)$. 
For every $a < b$, we can select a set of 
independent contrasts $\contrasts$,
associated with internal nodes 
in $\internal(a,b)$, such that
\[
|\contrasts| \geq n(a,b)/2,
\]
where $n(a,b) = |\internal(a,b)| $.
In particular, the heights of the contrasts in $\contrasts$
lie in $(a,b)$ and their corresponding
paths are pairwise non-intersecting.
\label{lem:choose}
\end{lemma}
\begin{proof}
Start with the lowest vertex $i$ in $\internal_{(a,b)}$
and choose a pair of vertices $\ell^i_1$ and
$\ell^i_2$ such that $\ell^i_1 \land \ell^i_2 = i$.
Remove $i$ and its descendants as well as the edge 
immediately above $i$ (and fuse consecutive
edges separated by degree-$2$ vertices). 
As a result, the number of
internal
vertices in $(a,b)$ decreases by at most $2$. Repeat
until no vertex is left in $\internal_{(a,b)}$.
\end{proof}

\paragraph{The estimator}
For a sequence of trees $\treeseq = (\tree_k)_k$,
let $\leaves_k$ be the leaf set of $\tree_k$;
$\internal_k$, the set of 
its internal vertices; $n_k = |\leaves_k|$ 
and $n_k(a,b) = |\internal_k(a,b)|$; 
and $T_k(i)$, the height of $i$, for each $i\in\internal_k$.
The idea behind our estimator is to set up
a system of equations that characterize $\alpha$
and $\gamma$ uniquely. 
Our construction relies on the following condition.
We illustrate this condition on two special cases
below.

We set up our equations as follows.
Let $m_k = \lfloor \beta n_k/2 \rfloor$.
Under Assumption~\ref{assump:alpha},
by Lemma~\ref{lem:choose}, for each $k$
we can choose {\em two} collections of 
independent contrasts $(\contrast^k_{i_r})_{r=1}^{m_k}$ 
and $(\contrast^k_{j_r})_{r=1}^{m_k}$ 
with corresponding heights 
$T_k(i_r) \in (c_1,c_1')$ 
and $T_k(j_r) \in (c_2,c_2')$ 
for every $r = 1, 2, \ldots, m_k$.
(Note that the two collections are {\em not}
independent.)
For $r = 1,\ldots,m$, let
\begin{eqnarray*}
\hat a_{k} 
= \frac{1}{m_k} \sum_{r=1}^{m_k}{(\contrast^{(k)}_{i_r})^2},
\qquad
\hat b_{k} 
= \frac{1}{m_k} \sum_{r=1}^{m_k}{(\contrast^{(k)}_{j_r})^2},
\end{eqnarray*}
and note that
\begin{eqnarray*}
a_{k} &\equiv& \E{\hat a_{k}} = 2 \gamma \left (1 - \frac{1}{m_k} \sum_{r=1}^{m_k}{e^{- 2 \alpha T_k(i_r)}} \right )
\equiv 2\gamma h^1_{k}(\alpha),\\
b_{k} &\equiv& \E{\hat b_{k}} = 2 \gamma \left (1 - \frac{1}{m_k} \sum_{r=1}^{m_k}{e^{- 2 \alpha T_k(j_r)}} \right )
\equiv 2\gamma h^2_{k}(\alpha).
\end{eqnarray*}
Notice that, under Assumption~\ref{assump:alpha}, 
$a_{k} \in [2\gamma(1-e^{-2 \alpha c_2}),2\gamma(1-e^{-2 \alpha c_1})] \equiv [\underline{a}_\alpha,\bar{a}_\alpha]$ and $b_{k} \in [2\gamma(1-e^{-2 \alpha c_4}),2\gamma(1-e^{-2 \alpha c_3})] \equiv [\underline{b}_\alpha,\bar{b}_\alpha]$.
As shown below,
$$
H_k(\alpha) = \frac{a_k}{b_k} = \frac{h^1_{k}(\alpha)}
{h^2_{k}(\alpha)}
$$
is invertible in $\alpha$ on $(0,+\infty)$. Hence
a natural estimator of $(\alpha,\gamma)$ is obtained
by setting
\begin{equation*} 
%\label{eq:hat-alpha-gamma}
\hat\alpha_k = H_k^{-1}\left(\frac{\hat a_k}{\hat b_k}\right) 
\quad\mbox{and}\quad
\hat\gamma_k = \frac{\hat a_k}{2 h_k^1(\hat\alpha_k)}.
\end{equation*}
We will show in the proof of invertibility below
that $H_k$ is actually strictly increasing, 
and therefore relatively straightforward to invert numerically.
It remains to prove invertibility.
\begin{lemma}[Invertibility of the system]
\label{lemma:invertibility}
Under Assumption~\ref{assump:alpha},
$H_{k}(\alpha)$ is strictly positive, 
differentiable, and invertible on $(0,+\infty)$.
\end{lemma} 
\begin{proof}
We have that
\begin{eqnarray}
\frac{\partial \log H_{k}(\alpha)}{\partial \alpha} 
&=& \frac{\sum_{r=1}^{m_k} {2T_k(i_r)}e^{-2 \alpha T_k(i_r)}}{\sum_{r=1}^{m_k}{(1 - e^{-2 \alpha T_k(i_r)}})} - \frac{\sum_{r=1}^{m_k}{2T_k(j_r) e^{-2 \alpha T_k(j_r)}}}{\sum_{r=1}^{m_k}{(1 - e^{-2 \alpha T_k(j_r)})}}\nonumber\\
&=& \frac{\sum\sum_{r, r'=1}^{m_k} {2T_k(i_r)}e^{-2 \alpha T_k(i_r)}
{(1 - e^{-2 \alpha T_k(j_{r'})})}}{\sum_{r=1}^{m_k}{(1 - e^{-2 \alpha T_k(i_r)}}) \sum_{r=1}^{m_k}{(1 - e^{-2 \alpha T_k(j_r)})}}\nonumber\\
&& \qquad - \frac{\sum\sum_{r,r'=1}^{m_k}{2 T_k(j_{r'}) e^{-2 \alpha T_k(j_{r'})}}(1 - e^{-2 \alpha T_k(i_{r})})}{\sum_{r=1}^{m_k}{(1 - e^{-2 \alpha T_k(i_r)}})\sum_{r=1}^{m_k}{(1 - e^{-2 \alpha T_k(j_r)})}}
\label{eq:invertibility-1}
\end{eqnarray}
Note that the function $\frac{xe^{-x}}{1 - e^{-x}}$ is strictly decreasing on $(0, \infty)$ because its derivative is $\frac{e^{-x}(1 - x - e^{-x})}{(1 - e^{-x})^2} < 0$ on $(0,+\infty)$. Therefore 
\[
\frac{2T_k(i_r)e^{-2 \alpha T_k(i_r)}}{1 - e^{-2 \alpha T_k(i_r)}} \geq 
\frac{2 c_1' e^{-2 \alpha c_1'}}{1 - e^{-2 \alpha c_1'}} >
\frac{2 c_2 e^{-2 \alpha c_2 }}{1 - e^{-2 \alpha c_2}} \geq
\frac{2T_k(j_{r'}) e^{-2 \alpha T_k(j_{r'})}}{1 - e^{-2 \alpha T_k(j_{r'})}},
\]
that is,
\begin{eqnarray}
&&2T_k(i_r)e^{-2 \alpha T_k(i_r)} (1 - e^{-2 \alpha T_k(j_{r'})})\nonumber\\
&& \qquad - 2T_k(j_{r'}) e^{-2 \alpha T_k(j_{r'})} (1 - e^{-2 \alpha T_k(i_r)})
> 0,\label{eq:logh-difference}
\end{eqnarray}
for every $r,r'$, so that each $(r,r')$-term in~\eqref{eq:invertibility-1}
is strictly positive.
Hence, we can deduce that 
$\partial \log H_k(\alpha)/\partial \alpha > 0$, that is, 
$\log H_k$ (and hence $H_k$ itself) is strictly increasing on $(0,+\infty)$
and continuous, and therefore invertible.
\end{proof}

Note that we cannot use the law of large numbers 
to derive consistency (despite
the independence of the contrasts)
because $a_k/b_k$ is a bounded, but
not necessarily convergent, sequence
and $H_k^{-1}$ is continuous, but
depends on $k$. Instead we argue directly
about $\sqrt{n_k}$-consistency below.

\subsection{Proof of Theorem~\ref{thm:main-alpha-rate}}
\label{sec:alpha-rate}

%We prove here Theorem~\ref{thm:main-alpha-rate} on the $\sqrt{n_k}$-consistency
%of the estimators $\hat\alpha_k$ and $\hat\gamma_k$ defined in
%\eqref{eq:hat-alpha-gamma}.
%
\begin{proof}[Proof of Theorem~\ref{thm:main-alpha-rate}]
Note that $\E{\hat a_k} = a_k$ and
$$
\var[\hat a_k]
= \frac{8\gamma^2}{m_k^2}
\sum_{r=1}^{m_k} 
(1-e^{-2\alpha T_k(i_r)})^2
\leq \frac{8\gamma^2}{m_k}
(1-e^{-2\alpha c_1})^2 
= O(m_k^{-1})
= O(n_k^{-1}),
$$
where we used that $([2\gamma(1-e^{-2\alpha T_k(i_r)})]^{-1/2} \contrast^k_{i_r})^2 $ is $\chi^2_1$-distributed and, therefore, has
variance $2$. Hence %by Lemma~\ref{lem01},
$|\hat a_k - a_k| = O_p(n_k^{-1/2})$ by Chebyshev's inequality.
Similarly, $|\hat b_k - b_k| = O_p(n_k^{-1/2})$.
Our claim that $|\hat\alpha_k - \alpha_k| = O_p(n_k^{-1/2})$
then follows from the following 
straightforward lemma.
\begin{lemma}
If $0 < z_* \leq z \leq z^* < \infty$, $|z' - z| \leq \epsilon$ and $\epsilon < z_*/2$, then there is a constant $\Delta(z_*,z^*)$ depending
on $c_1,c_1',c_2,c_2'$ such that for all $k$
\[
\sup_{t \in [0,1]}{|(H_k^{-1})' (t z' + (1-t)z )|} \leq \Delta(z_*,z^*).
\]
\label{lem:monotone} 
\end{lemma}
\begin{proof}
We use the proof of Lemma~\ref{lemma:invertibility}.
Let $\zeta_\alpha = \zeta_\alpha(c_1,c_1',c_2,c_2') > 0$ be the smallest possible difference in~\eqref{eq:logh-difference}
for a fixed $\alpha$. Let $\alpha_*, \alpha^*$ be defined
as 
$$
\frac{1}{2}z_* = \frac{\bar{a}_{\alpha_*}}{\underline b_{\alpha_*}},
\qquad \frac{3}{2}z^* = \frac{\underline{a}_{\alpha^*}}{\bar b_{\alpha^*}}.
$$
Then
$[\alpha_*,\alpha^*]
\supseteq H_k^{-1}\left(\left[\frac{1}{2}z_*,\frac{3}{2}z^*\right]\right)$ for all $k$.
Note that
\begin{eqnarray*}
\sup_{t \in [0,1]}{|(H_k^{-1})' (t z' + (1-t)z )|}
&\leq& \sup_{z\in\left[\frac{1}{2}z_*,\frac{3}{2}z^*\right]}\left|(H_k^{-1})'(z)\right|\\
&=& \sup_{z\in\left[\frac{1}{2}z_*,\frac{3}{2}z^*\right]}\left(\frac{\partial H_k}{\partial \alpha} (H_k^{-1}(z))\right)^{-1}\\
&=& \sup_{z\in\left[\frac{1}{2}z_*,\frac{3}{2}z^*\right]}\left(\left[H_k\frac{\partial \log H_k}{\partial \alpha}\right] (H_k^{-1}(z))\right)^{-1}\\
&\leq& \sup_{\alpha\in[\alpha_*,\alpha^*]}\frac{\underline{b}_\alpha}{\bar a_\alpha}
\cdot \frac{(1-e^{-2\alpha c_1'})(1 - e^{-2\alpha c_2'})}{\zeta_\alpha}\\
&\equiv& \Delta(z_*,z^*).
\end{eqnarray*}
\end{proof}

We finish the proof of Theorem~\ref{thm:main-alpha-rate}.
We use the following observation: for $0 < x_* \leq x \leq x^* < \infty $ and $0 < y_* \leq y \leq y^* < \infty$ such that 
$|x - x'| \leq \epsilon$ and $|y - y'| \leq \epsilon$ 
with $\epsilon < y_*/2$, we have
\begin{equation*}
%\label{eq:old-lemma}
\left | \frac{x'}{y'} - \frac{x}{y} \right |
= \left | \frac{y(x' - x) + x(y - y')}{y y'} \right | \leq \frac{y^*|x' - x|}{y_*(y_*/2)} + \frac{x^*|y - y'|}{y_*(y_*/2)}
 < \frac{4(x^* + y^*)}{y_*^2}\epsilon.
\end{equation*}
Fix $\delta > 0$ (small) and pick $M_\delta$
such that $\P{|\hat a_k - a_k|\geq M_\delta n_k^{-1/2}} <\delta/2$
and similarly for $\hat b_k$. Then,
by Assumption~\ref{assumption:unbounded}, for $k$ large enough
{\small\begin{eqnarray*}
&&\P{\left | \frac{\hat a_k}{\hat b_k} - \frac{a_k}{b_k} \right | 
\geq \frac{4(\bar a_\alpha + \bar b_\alpha)}{\underline{b}_\alpha^2} M_\delta n_k^{-1/2}}\\ 
&&\qquad \leq \P{ \left |\frac{\hat a_k}{\hat b_k} - \frac{a_k}{b_k} \right | \geq \frac{4(\bar a_\alpha + \bar b_\alpha)}{\underline{b}_\alpha^2} M_\delta n_k^{-1/2}, |\hat a_k - a_k| \leq M_\delta n_k^{-1/2}, |\hat b_k - b_k| \leq M_\delta n_k^{-1/2}}\\
&& \qquad\qquad  + \P{|\hat a_k - a_k| \geq M_\delta n_k^{-1/2}} 
+ \P{|\hat b_k - b_k| \geq M_\delta n_k^{-1/2}}\\
&&\qquad \leq 0 + \frac{\delta}{2} + \frac{\delta}{2} = \delta,
\end{eqnarray*}}
so that 
$
\left | \frac{\hat a_k}{\hat b_k} - \frac{a_k}{b_k} \right |
= O_p(n_k^{-1/2}).
$

Secondly, using Rolle's theorem, we have 
\[
|\hat \alpha_k - \alpha| 
\leq \sup_{t \in [0,1]}{\left |(H_k^{-1})' \left (t\frac{\hat a_k}{\hat b_k} + (1-t)\frac{a_k}{b_k} \right ) \right|. \left |\frac{\hat a_k}{\hat b_k} - \frac{a_k}{b_k} \right |}.
\]
Let 
$M_\delta$ be 
such that 
$$
\P{\left|\frac{\hat a_k}{\hat b_k} - \frac{a_k}{b_k}\right|\geq M_\delta n_k^{-1/2}} <\delta.
$$
Fix $\epsilon' > 0$ and let 
$$
z_* = \frac{\underline{a}_{\alpha - \epsilon'}}{\bar{b}_{\alpha - \epsilon'}}, \qquad 
z^* = \frac{\bar{a}_{\alpha + \epsilon'}}{\underline{b}_{\alpha + \epsilon'}}.
$$
Then, by Lemma \ref{lem:monotone}, 
letting
$$
\mathcal{H}_k
= \left\{
\sup_{t \in [0,1]}{\left |(H_k^{-1})' \left (t\frac{\hat a_k}{\hat b_k} + (1-t)\frac{a_k}{b_k} \right ) \right|. 
\left |\frac{\hat a_k}{\hat b_m} - \frac{a_m}{b_m} \right |} \geq \Delta^{-1}(z_*,z^*) M_\delta n_k^{-1/2}
\right\},
$$
we have for $k$ large enough
\begin{eqnarray*}
&&\P{|\hat \alpha_k - \alpha| \geq \Delta^{-1}(z_*,z^*) M_\delta n_k^{-1/2}}\\
&&\quad\leq \P{\mathcal{H}_k}\\
&&\quad\leq \P{\mathcal{H}_k, \left |\frac{\hat a_k}{\hat b_k} - \frac{a_k}{b_k} \right | < M_\delta n_k^{-1/2}}
+ \P{\left |\frac{\hat a_k}{\hat b_k} - \frac{a_k}{b_k} \right | \geq M_\delta n_k^{-1/2}}\\
&&\quad\leq 0 + \delta 
= \delta.
\end{eqnarray*}
That implies
$ |\hat\alpha_k - \alpha|= O_p(n_k^{-1/2}). $
The argument for $\hat\gamma_k$ is similar.
%To deal with the denominator, note that
%\begin{eqnarray*}
%&& \P{\left|2h^1_k(\hat\alpha_k) - 2h^1_k(\alpha)\right|
%\geq M_\delta n_k^{-1/2}}\\
%&&\qquad \leq \P{2\left | \frac{1}{m_k}\sum_{r=1}^{m_k}{\left ( e^{-2 \hat\alpha_k T_k(i_r)} - e^{-2 \alpha T_k(i_r)} \right )} \right | \geq M_\delta n_k^{-1/2}}\\ 
%%&&\qquad \leq \P{2\frac{1}{m_k}\sum_{r=1}^{m_k}{\left | e^{-2 \hat\alpha_k T_k(i_r)} - e^{-2 \alpha T_k(i_r)} \right |} \geq M_\delta n_k^{-1/2}}\\
%&&\qquad \leq \P{2\frac{1}{m_k}\sum_{r=1}^{m_k}{2T_k(i_r) \left |\hat \alpha_k - \alpha \right |} \geq M_\delta n_k^{-1/2}}\\ 
%&&\qquad \leq \P{4 c_1 |\hat \alpha_k - \alpha| \geq M_\delta n_k^{-1/2}},
%\end{eqnarray*}
%where we used that $|e^{-x} - e^{-y}| \leq |x-y|$ for $x,y \geq 0$.
%Then use Lemma~\ref{lem:ratio}
%as above.
\end{proof}

%
%\begin{remark}
%The previous result can be further generalized to
%sublinear-sized collections of contrasts whose height grows
%with $k$ (with a different rate of convergence). 
%We leave out the details. 
%\end{remark}

%\begin{acknowledgements}
%\end{acknowledgements}

\bibliographystyle{spmpsci}      % mathematics and physical sciences

%\bibliography{ms}

\begin{thebibliography}{10}
\providecommand{\url}[1]{{#1}}
\providecommand{\urlprefix}{URL }
\expandafter\ifx\csname urlstyle\endcsname\relax
  \providecommand{\doi}[1]{DOI~\discretionary{}{}{}#1}\else
  \providecommand{\doi}{DOI~\discretionary{}{}{}\begingroup
  \urlstyle{rm}\Url}\fi

\bibitem{adamczak2011clt}
Adamczak, R., Mi{\l}o{\'s}, P.: {CLT} for {O}rnstein-{U}hlenbeck branching
  particle system.
\newblock Electronic Journal of Probability \textbf{20}(42), 1--35 (2015)

\bibitem{adamczak2011u}
Adamczak, R., Mi{\l}o{\'s}, P.: {U}-{S}tatistics of {O}rnstein-{U}hlenbeck
  branching particle system.
\newblock Journal of Theoretical Probability \textbf{27}(4), 1071--1111 (2014)

\bibitem{Anderson:1984}
Anderson, T.W.: An introduction to multivariate statistical analysis, 2nd edn.
\newblock Wiley, Chichester (1984)

\bibitem{athreya2004branching}
Athreya, K., Ney, P.: Branching Processes.
\newblock Dover Books on Mathematics Series. Dover Publications (2004)

\bibitem{bartoszek2012multivariate}
Bartoszek, K., Pienaar, J., Mostad, P., Andersson, S., Hansen, T.F.: A
  phylogenetic comparative method for studying multivariate adaptation.
\newblock Journal of Theoretical Biology \textbf{314}, 204--215 (2012)

\bibitem{bartoszek2012phylogenetic}
Bartoszek, K., Sagitov, S.: Phylogenetic confidence intervals for the optimal
  trait value.
\newblock Journal of Applied Probability  \textbf{52}(4), 1115--1132 (2015).

\bibitem{binindaEmonds-etal07}
Bininda-Emonds, O., Cardillo, M., Jones, K.E., MacPhee, R.D.E., Beck, R.M.D.,
  Grenyer, R., Price, S.A., Vos, R.A., Gittleman, J.L., Purvis, A.: The delayed
  rise of present-day mammals.
\newblock Nature \textbf{446}(7135), 507--512 (2007)

\bibitem{brawand_etal11}
Brawand, D., Soumillon, M., Necsulea, A., Julien, P., Csardi, G., Harrigan, P.,
  Weier, M., Liechti, A., Aximu-Petri, A., Kircher, M., Albert, F.W., Zeller,
  U., Khaitovich, P., Grutzner, F., Bergmann, S., Nielsen, R., P\"a\"abo, S.,
  Kaessmann, H.: The evolution of gene expression levels in mammalian organs.
\newblock Nature \textbf{478}(7369), 343--348 (2011)

\bibitem{butler2004phylogenetic}
Butler, M.A., King, A.A.: Phylogenetic comparative analysis: a modeling
  approach for adaptive evolution.
\newblock The American Naturalist \textbf{164}(6), 683--695 (2004)

\bibitem{cooperPurvis10}
Cooper, N., Purvis, A.: Body size evolution in mammals: Complexity in tempo and
  mode.
\newblock The American Naturalist \textbf{175}(6), 727--738 (2010)

\bibitem{crawford2013diversity}
Crawford, F.W., Suchard, M.A.: Diversity, disparity, and evolutionary rate
  estimation for unresolved {Y}ule trees.
\newblock Systematic Biology \textbf{62}(3), 439--455 (2013)

\bibitem{EvKePeSc:00}
Evans, W.S., Kenyon, C., Peres, Y., Schulman, L.J.: Broadcasting on trees and
  the {I}sing model.
\newblock Ann. Appl. Probab. \textbf{10}(2), 410--433 (2000)

\bibitem{felsenstein1985phylogenies}
Felsenstein, J.: Phylogenies and the comparative method.
\newblock American Naturalist \textbf{125}(1), 1--15 (1985)

\bibitem{felsenstein2004inferring}
Felsenstein, J.: Inferring Phylogenies.
\newblock Sinauer Associates (2004)

\bibitem{hansen1997stabilizing}
Hansen, T.F.: Stabilizing selection and the comparative analysis of adaptation.
\newblock Evolution \textbf{51}(5), 1341--1351 (1997)

\bibitem{geiger}
Harmon, L., Weir, J., Brock, C., Glor, R., Challenger, W.: G{EIGER}:
  investigating evolutionary radiations.
\newblock Bioinformatics \textbf{24}, 129--131 (2008)

\bibitem{harmon-etal10}
Harmon, L.J., Losos, J.B., Jonathan~Davies, T., Gillespie, R.G., Gittleman,
  J.L., Bryan~Jennings, W., Kozak, K.H., McPeek, M.A., Moreno-Roark, F., Near,
  T.J., Purvis, A., Ricklefs, R.E., Schluter, D., Schulte~II, J.A., Seehausen,
  O., Sidlauskas, B.L., Torres-Carvajal, O., Weir, J.T., Mooers, A.\O.: Early
  bursts of body size and shape evolution are rare in comparative data.
\newblock Evolution \textbf{64}(8), 2385--2396 (2010)

\bibitem{HoAne13}
Ho, L.S.T., An\'e, C.: Asymptotic theory with hierarchical autocorrelation:
  {Ornstein-Uhlenbeck} tree models.
\newblock Annals of Statistics \textbf{41}, 957--981 (2013)

\bibitem{ho2014intrinsic}
Ho, L.S.T., An{\'e}, C.: Intrinsic inference difficulties for trait evolution
  with {Ornstein-Uhlenbeck} models.
\newblock Methods in Ecology and Evolution \textbf{5}(11), 1133--1146 (2014)

\bibitem{HoAne14sb}
Ho, L.S.T., An\'e, C.: A linear-time algorithm for {G}aussian and
  non-{G}aussian trait evolution models.
\newblock Systematic Biology \textbf{63}(3), 397--408 (2014)

\bibitem{jetz12}
Jetz, W., Thomas, G., Joy, J., Hartmann, K., Mooers, A.: The global diversity
  of birds in space and time.
\newblock Nature \textbf{491}(7424), 444--448 (2012)

\bibitem{lawler1976combinatorial}
Lawler, E.: Combinatorial Optimization: Networks and Matroids.
\newblock Holt, Rinehart and Winston (1976)

\bibitem{mossel2013robust}
Mossel, E., Roch, S., Sly, A.: Robust estimation of latent tree graphical
  models: Inferring hidden states with inexact parameters.
\newblock IEEE transactions on information theory \textbf{59}(7), 4357--4373
  (2013)

\bibitem{mossel2014majority}
Mossel, E., Steel, M.: Majority rule has transition ratio 4 on yule trees under
  a 2-state symmetric model.
\newblock Journal of Theoretical Biology \textbf{360}(7), 315--318 (2014).

\bibitem{ape}
Paradis, E., Claude, J., Strimmer, K.: A{PE}: analyses of phylogenetics and
  evolution in {R} language.
\newblock Bioinformatics \textbf{20}, 289--290 (2004)

\bibitem{peres1999climb}
Peres, Y.: Probability on trees: An introductory climb.
\newblock In: P.~Bernard (ed.) Lectures on Probability Theory and Statistics,
  \emph{Lecture Notes in Mathematics}, vol. 1717, 193--280. Springer Berlin
  Heidelberg (1999)

\bibitem{rohlfs2014modeling}
Rohlfs, R.V., Harrigan, P., Nielsen, R.: Modeling gene expression evolution
  with an extended {O}rnstein-{U}hlenbeck process accounting for within-species
  variation.
\newblock Molecular Biology and Evolution \textbf{31}(1), 201--211 (2014)

\bibitem{semple2003phylogenetics}
Semple, C., Steel, A.: Phylogenetics.
\newblock Oxford lecture series in mathematics and its applications. Oxford
  University Press (2003)

\bibitem{shao2003mathstat}
Shao, J.: Mathematical Statistics.
\newblock Springer (2003)

\bibitem{vendittiMeadePagel11}
Venditti, C., Meade, A., Pagel, M.: Multiple routes to mammalian diversity.
\newblock Nature \textbf{479}(7373), 393--396 (2011)

\bibitem{yule1925mathematical}
Yule, G.U.: A mathematical theory of evolution, based on the conclusions of
  {Dr. JC Willis, FRS}.
\newblock Philosophical Transactions of the Royal Society of London. Series B
  \textbf{213}, 21--87 (1925)

\end{thebibliography}

\end{document}